\documentclass[a4paper,10pt]{amsart}

\usepackage{hyperref}
\usepackage{amsthm,latexsym,amssymb,amsmath,stmaryrd}
\usepackage{pifont}
\usepackage{times,mathptmx}
\usepackage{enumerate}
\usepackage{leftidx}
\usepackage[all]{xy}

\usepackage{xcolor}
\usepackage{pstricks,epsf}

\newcommand{\A}{\mathcal{A}}
\newcommand{\B}{\mathcal{B}}
\newcommand{\CC}{\mathcal{C}}
\newcommand{\CM}{\mathcal{C}_M}
\newcommand{\OO}{\mathcal{O}}
\def\G{\mathcal{G}}
\renewcommand{\P}{\mathcal{P}}
\newcommand{\N}{\mathbb{N}}

\newcommand{\R}{\mathbb{R}}
\def\C{\mathbb{C}}
\newcommand{\K}{\mathbb{K}}
\newcommand{\Pt}{\ltrans{P}}
\newcommand{\<}{\langle}
\renewcommand{\>}{\rangle}
\renewcommand{\phi}{\varphi}
\newcommand{\na}{\nabla}

\newcommand{\eps}{\varepsilon}
\newcommand{\dV}{\,\mathrm{dV}}
\newcommand{\vol}{\,\mathrm{vol}}
\newcommand{\dA}{\,\mathrm{dA}}
\newcommand{\dAs}{\,\mathrm{dA}_s}
\newcommand{\supp}{\mathrm{supp}}
\newcommand{\ext}{\mathrm{ext}}
\newcommand{\res}{\mathrm{res}}
\newcommand{\Hom}{\mathrm{Hom}}
\renewcommand{\Re}{\mathrm{Re}}

\newcommand{\Coo}{C^\infty}
\newcommand{\CooA}{C_A^\infty}
\newcommand{\Cooc}{C_c^\infty}
\newcommand{\Coosc}{C_{sc}^\infty}
\newcommand{\Cootc}{C_{tc}^\infty}
\newcommand{\Coospc}{C_{spc}^\infty}
\newcommand{\Coosfc}{C_{sfc}^\infty}
\newcommand{\Coopc}{C_{pc}^\infty}
\newcommand{\Coofc}{C_{fc}^\infty}

\newcommand{\D}{\mathcal{D}'}
\newcommand{\Dc}{\mathcal{D}_c'}
\newcommand{\Dfc}{\mathcal{D}_{fc}'}
\newcommand{\Dpc}{\mathcal{D}_{pc}'}
\newcommand{\Dspc}{\mathcal{D}_{spc}'}
\newcommand{\Dsfc}{\mathcal{D}_{sfc}'}
\newcommand{\Dsc}{\mathcal{D}_{sc}'}
\newcommand{\Dtc}{\mathcal{D}_{tc}'}

\newtheorem{thm}{Theorem}[section]

\newtheorem{cor}[thm]{Corollary}
\newtheorem{prop}[thm]{Proposition}
\newtheorem{lem}[thm]{Lemma}

\theoremstyle{definition}
\newtheorem{ex}[thm]{Example}
\newtheorem{rem}[thm]{Remark}

\newtheorem{dfn}[thm]{Definition}

\newcommand{\dref}[1]{Definition~\ref{#1}}
\newcommand{\tref}[1]{Theorem~\ref{#1}}
\newcommand{\pref}[1]{Proposition~\ref{#1}}
\newcommand{\lref}[1]{Lemma~\ref{#1}}
\newcommand{\cref}[1]{Corollary~\ref{#1}}

\newcommand{\eref}[1]{Example~\ref{#1}}
\newcommand{\rref}[1]{Remark~\ref{#1}}
\newcommand{\diagref}[1]{Diagram~\ref{#1}}

\newcounter{diagram}
\newenvironment{diagram}[0]{\refstepcounter{diagram}}{} 
\newcommand{\diagramnumber}[1]{
\vspace{-8pt}
\begin{center}
{\sc Diagram}~\arabic{diagram}: \emph{#1}
\end{center}
}

\title{Green-hyperbolic operators on globally hyperbolic spacetimes}
\author{Christian B\"ar}
\address{Universit\"at Potsdam, Institut f\"ur Mathematik, Am Neuen Palais 10, 14469 Potsdam, Germany}
\email{baer@math.uni-potsdam.de}
\urladdr{http://geometrie.math.uni-potsdam.de/}

\keywords{Globally hyperbolic Lorentzian manifolds, Green-hyperbolic operators, wave operators, normally hyperbolic operators, Dirac-type operators, Green's operators, support system, symmetric hyperbolic system, Cauchy problem, energy estimate, finite propagation speed, locally covariant quantum field theory}
\subjclass[2010]{58J45,35L45,35L51,35L55,81T20}

\begin{document}

\begin{abstract}
Green-hyperbolic operators are linear differential operators acting on sections of a vector bundle over a Lorentzian manifold which possess advanced and retarded Green's operators.
The most prominent examples are wave operators and Dirac-type operators.
This paper is devoted to a systematic study of this class of differential operators.
For instance, we show that this class is closed under taking restrictions to suitable subregions of the manifold, under composition, under taking ``square roots'', and under the direct sum construction.
Symmetric hyperbolic systems are studied in detail.
\end{abstract}

\maketitle

%%%%%%%%%%%%%%%%%%%%%%%%%%%%%%%%%%%%%%%%%%%%%%%%%%%%%%%%%%%%%%%%%%%%%%%%%%%%%%%%%%%%%%%%%%%%%%%%%

\section*{Introduction}

Green-hyperbolic operators are certain linear differential operators acting on sections of a vector bundle over a Lorentzian manifold.
They are, by definition, those operators which possess advanced and retarded Green's operators.
The most prominent examples are normally hyperbolic operators (wave equations) and Dirac-type operators.
The reason for introducing them in \cite{BG12b} lies in the fact that they can be quantized;
one can canonically construct a bosonic locally covariant quantum field theory for them.

The aim of the present paper is to study Green-hyperbolic operators systematically from a geometric and an analytic perspective.
The underlying Lorentzian manifold must be well behaved for the analysis of hyperbolic operators.
In technical terms, it must be globally hyperbolic.
In the first section we collect material about such Lorentzian manifolds.
We introduce various compactness properties for closed subsets and show their interrelation.
These considerations will later be applied to the supports of sections.

In the second section we study various spaces of smooth sections of our vector bundle.
The crucial concept is that of a support system.
This is a family of closed subsets of our manifold with certain properties making it suitable for defining a good space of sections by demanding that their supports be contained in the support system.
We observe a duality principle;
a distributional section has support in a support system if and only if it extends to a continuous linear functional on test sections with support in the dual support system.

Green's operators and Green-hyperbolic differential operators are introduced in the third section.
We give various examples and show that the class of Green-hyperbolic operators is closed under taking restrictions to suitable subregions of the manifold, under composition, under taking ``square roots'', and under the direct sum construction.
This makes it a large and very flexible class of differential operators to consider.
We show that the Green's operators are unique and that they extend to several spaces of sections.
We argue that Green-hyperbolic operators are not necessarily hyperbolic in any PDE-sense and that they cannot be characterized in general by well-posedness of a Cauchy problem.

The fourth section is devoted to extending the Green's operators to distributional sections.
We show that an important analytical result for the causal propagator (the difference of the advanced and the retarded Green's operator), also holds when one replaces smooth by distributional sections.

In the last section we study symmetric hyperbolic systems over globally hyperbolic manifolds.
We provide detailed proofs of well-posedness of the Cauchy problem, finiteness of the speed of propagation and the existence of Green's operators.
The crucial step in these investigations is an energy estimate for the solution to such a symmetric hyperbolic system.
We conclude by observing that a symmetric hyperbolic system can be quantized in two ways;
one yields a \emph{bosonic} and the other one a \emph{fermionic} locally covariant quantum field theory.

\emph{Acknowledgments.}
It is a great pleasure to thank Klaus Fredenhagen, Ulrich Menne, and Miguel S\'anchez for very helpful discussions and an anonymous referee for very interesting suggestions.
Thanks also go to Sonderforschungs\-bereich~647 funded by Deutsche Forschungsgemeinschaft for financial support.

%%%%%%%%%%%%%%%%%%%%%%%%%%%%%%%%%%%%%%%%%%%%%%%%%%%%%%%%%%%%%%%%%%%%%%%%%%%%%%%%%%%%%%%%%%%%%%%%%

\section{Globally Hyperbolic Lorentzian Manifolds}

We summarize various facts about globally hyperbolic Lorentzian manifolds. 
For details the reader is referred to one of the classical textbooks \cite{BEE96,HE73,ON83}.
Throughout  this article, $M$ will denote a time oriented Lorentzian manifold.
We use the convention that the signature of $M$ is $(-+\cdots+)$.
Note that we do not specify the dimension of $M$ nor do we assume orientability or connectedness.

\subsection{Cauchy hypersurfaces}
A subset $\Sigma\subset M$ is called a \emph{Cauchy hypersurface} if every inextensible timelike curve in $M$ meets $\Sigma$ exactly once.
Any Cauchy hypersurface is a topological submanifold of codimension $1$.
All Cauchy hypersurfaces of $M$ are homeomorphic.

If $M$ possesses a Cauchy hypersurface then $M$ is called \emph{globally hyperbolic}.
This class of Lorentzian manifolds contains many important examples:
Minkowski space, Friedmann models, the Schwarzschild model and deSitter spacetime are globally hyperbolic.
Bernal and S\'anchez proved an important structural result \cite[Thm.~1.1]{BS05}:
Any globally hyperbolic Lorentzian manifold has a \emph{Cauchy temporal function}.
This is a smooth function $t:M\to\R$ with past-directed timelike gradient $\na t$ such that the levels $t^{-1}(s)$ are (smooth spacelike) Cauchy hypersurfaces if nonempty.

\subsection{Future and past}
From now on let $M$ always be globally hyperbolic.
For any $x\in M$ we denote by $J^+(x)$ the set all points that can be reached by future-directed causal curves emanating from $x$.
For any subset $A\subset M$ we put $J^+(A):=\bigcup_{x\in A}J^+(x)$.
If $A$ is closed so is $J^+(A)$.
We call a subset $A\subset M$ \emph{strictly past compact} if it is closed and there is a compact subset $K\subset M$ such that $A\subset J^+(K)$.
If $A$ is strictly past compact so is $J^+(A)$ because $J^+(A)\subset J^+(J^+(K))=J^+(K)$.
A closed subset $A\subset M$ is called \emph{future compact} if $A\cap J^+(x)$ is compact for all $x\in M$.
For example, if $\Sigma$ is a Cauchy hypersurface, then $J^-(\Sigma)$ is future compact.

We denote by $I^+(x)$ the set of all points in $M$ that can be reached by future-directed timelike curves emanating from $x$.
The set $I^+(x)$ is the interior of $J^+(x)$; in particular, it is an open subset of $M$.

Interchanging the roles of future and past, we similarly define $J^-(x)$, $J^-(A)$, $I^-(x)$, \emph{strictly future compact} and \emph{past compact} subsets of $M$. 
If $A\subset M$ is past compact and future compact then we call $A$ \emph{temporally compact}.
For any compact subsets $K_1,K_2\subset M$ the intersection $J^+(K_1)\cap J^-(K_2)$ is compact.
If $A$ is past compact so is $J^+(A)$ because $J^+(A)\cap J^-(x)=J^+(A\cap J^-(x)) \cap J^-(x)$.
Similarly, if $A$ is future compact then $J^-(A)$ is future compact too.
If $A$ is strictly past compact then it is past compact because $A\cap J^-(x)\subset J^+(K)\cap J^-(x)$ is compact.
Similarly, strictly future compact sets are future compact.

If we want to emphasize the ambient manifold $M$, then we write $J^+_M(x)$ instead of $J^+(x)$ and similarly for $J^-_M(x)$, $J^\pm_M(A)$, and $I^\pm_M(A)$.

\begin{ex}\label{ex:pastcompactbutnotstrict}
Let $M$ be Minkowski space and let $C\subset M$ be an open cone with tip $0$ containing the closed cone $J^-(0)\setminus\{0\}$.
Then $A=M\setminus C$ is past compact but not strictly past compact.
Indeed, for each $x\in M$, the set $J^-(x)\cap A=J^-(x)\setminus C$ is compact.
But $A$ is not strictly past compact because the intersection of $A$ and spacelike hyperplanes is not compact, compare \lref{lem:SpacelikeCompact} below (Fig.~1).
\begin{center}
\begin{pspicture}(-4,-2.6)(4,0.8)
\psline[fillstyle=solid,fillcolor=white](-2.1,-1.9)(0.1,0.3)(2.9,-2.5)
\psline[fillstyle=solid,fillcolor=gray](-3,-2.5)(0,0)(3,-2.5)
\psline[fillstyle=solid,fillcolor=lightgray,linestyle=dotted](-2.5,-2.5)(0,0)(2.5,-2.5)
\psdots(0,0)(0.1,0.3)
\uput[180](0.1,0.3){$x$}
\uput[270](0,0){\psframebox*[framearc=.3]{$0$}}
\uput[0](1.1,-0.6){$J^-(x)$}
\uput[90](0,-2.4){\psframebox*[framearc=.3]{$J^-(0)$}}
\uput[180](-0.7,-1.4){\psframebox*[framearc=.3]{$C$}}
\end{pspicture}

{\sc Fig.}~1:
\emph{Past-compact set which is not strictly past compact}
\end{center}
This example also shows that (surjective) Cauchy temporal functions need not be bounded from below on past compact sets.
However, we have:
%The situation is different if the Cauchy hypersurfaces of $M$ are compact, see \lref{lem:pastcompact}.
\end{ex}

\begin{lem}\label{lem:CauchyChar1}
For any closed subset $A\subset M$ the following are equivalent:
\begin{enumerate}[(i)]
\item\label{pc:1}
$A$ is past compact;
\item\label{pc:2}
there exists a smooth spacelike Cauchy hypersurface $\Sigma\subset M$ such that $A\subset J^+(\Sigma)$;
\item\label{pc:3}
there exists a surjective Cauchy temporal function $t:M\to\R$ which is bounded from below on $A$.
\end{enumerate}
\end{lem}

\begin{proof}
The implication ``\eqref{pc:3} $\Rightarrow$ \eqref{pc:2}'' is trivial and the inverse implication is a consequence of \cite[Thm.~1.2]{BS06}.
The implication ``\eqref{pc:2} $\Rightarrow$ \eqref{pc:1}'' is also trivial because $J^+(\Sigma)$ is past compact.
We only need to show ``\eqref{pc:1} $\Rightarrow$ \eqref{pc:2}''.

Let $A$ be past compact.
Then $J^+(A)$ is also past compact.
Moreover, $M':=M\setminus J^+(A)$ is an open subset of $M$ with the property $J^-(M')=M'$.
Hence $M'$ is globally hyperbolic itself.
Let $\Sigma$ be a smooth spacelike Cauchy hypersurface in $M'$.
Since $A\subset J^+(A) \subset J^+(\Sigma)$ it remains to show that $\Sigma$ is also a Cauchy hypersurface in $M$.

Let $c$ be an inextensible future-directed timelike curve in $M$.
Once $c$ has entered $J^+(A)$ it remains in $J^+(A)$.
Since $J^+(A)$ is past compact and $c$ is inextensible, $c$ must also meet $M'$.
Thus $c$ is the concatenation of an inextensible future-directed timelike curve $c_1$ in $M'$ and a (possibly empty) curve $c_2$ in $J^+(A)$.
Since $c_1$ meets $\Sigma$ exactly once, so does $c$.
This shows that $\Sigma$ is a Cauchy hypersurface in $M$ as well.
\end{proof}

Reversing future and past, we see that a closed subset $A\subset M$ is future compact if and only if $A\subset J^-(\Sigma)$ for some Cauchy hypersurface $\Sigma\subset M$.
This in turn is equivalent to the existence of a surjective Cauchy temporal function $t:M\to\R$ which is bounded from above on $A$.

Consequently, $A$ is temporally compact if and only if $A\subset J^+(\Sigma_1) \cap J^-(\Sigma_2)$ for some Cauchy hypersurfaces $\Sigma_1,\Sigma_2\subset M$.

\begin{lem}
For any past-compact subset $A\subset M$ there exists a past-compact subset $A'\subset M$ such that $A$ is contained in the interior of $A'$.
Analogous statements hold for future-compact sets and for temporally compact sets.
\end{lem}

\begin{proof}
Let $A\subset M$ be past compact.
Choose a Cauchy hypersurface $\Sigma\subset M$ such that $A\subset J^+(\Sigma)$.
Choose a second Cauchy hypersurface $\Sigma' \subset I^-(\Sigma)$.
Then $A':=J^+(\Sigma')$ does the job.
\end{proof}

For $A\subset M$ we write $J(A):=J^+(A)\cup J^-(A)$.
We call $A$ \emph{spacially compact} if $A$ is closed and there exists a compact subset $K\subset M$ with $A\subset J(K)$.
We have the following analog to \lref{lem:CauchyChar1}:

\begin{lem}\label{lem:CauchyChar2}
For any closed subset $A\subset M$ the following holds:
\begin{enumerate}[(i)]
\item\label{spc:1}
$A$ is strictly past compact if and only if $A\subset J^+(K_\Sigma)$ for some compact subset $K_\Sigma$ of some smooth spacelike Cauchy hypersurface $\Sigma\subset M$;
\item\label{spc:2}
$A$ is strictly future compact if and only if $A\subset J^-(K_\Sigma)$ for some compact subset $K_\Sigma$ of some smooth spacelike Cauchy hypersurface $\Sigma\subset M$;
\item\label{spc:3}
$A$ is spacially compact if and only if $A\subset J(K_\Sigma)$ for some compact subset $K_\Sigma$ of \emph{any} Cauchy hypersurface $\Sigma\subset M$.
\end{enumerate}
\end{lem}

\begin{proof}
One direction in \eqref{spc:1} is trivial: 
if $A\subset J^+(K_\Sigma)$, then $A$ is strictly past compact by definition.
Conversely, let $A$ be strictly past compact and let $K\subset M$ be a compact subset such that $A\subset J^+(K)$.
Then choose a smooth spacelike Cauchy hypersurface $\Sigma\subset M$ such that $K\subset J^+(\Sigma)$ and put $K_\Sigma := \Sigma \cap J^-(K)$.
Then $K_\Sigma$ is compact and 
\[
A \subset J^+(K) \subset J^+(J^+(\Sigma)\cap J^-(K)) = J^+(\Sigma\cap J^-(K)) = J^+(K_\Sigma).
\]
The proof of \eqref{spc:2} is analogous.
As to \eqref{spc:3}, if $A$ is spacially compact and $\Sigma\subset M$ a Cauchy hypersurface, then $K_\Sigma:=\Sigma\cap J(K)$ does the job.
\end{proof}

We have the following diagram of implications of possible properties of a closed subset of $M$:
\begin{diagram}\label{diag:closed1}
\begin{equation*}
\xymatrix@R=4pt{
& \mbox{compact} \ar@{=>}[dl]\ar@{=>}[dr]& \\
\mbox{strictly past compact}\ar@{=>}[dr]\ar@{=>}[dd] && \mbox{strictly future compact}\ar@{=>}[dl]\ar@{=>}[dd] \\
& \mbox{spacially compact} & \\
\mbox{past compact} && \mbox{future compact} \\
& \mbox{temporally compact}\ar@{=>}[ul]\ar@{=>}[ur] & }
\end{equation*}

\vspace{1mm}
\diagramnumber{Possible properties of closed subsets}
\end{diagram}

\subsection{Spacially compact manifolds}
None of the reverse implications in the diagram holds in general.
In a special case however, the diagram simplifies considerably, see \rref{rem:Mspacelikecompact}.
The terminology ``spacially compact'' is justified by the following lemma:

\begin{lem}\label{lem:SpacelikeCompact}
Let $A\subset M$ be spacially compact and let $\Sigma\subset M$ be a Cauchy hypersurface.
Then $A\cap\Sigma$ is compact.
\end{lem}

\begin{proof}
For any $x\in M$ the intersection $J^-(x)\cap J^+(\Sigma)$ is compact by Lemma~40 in \cite[p.~423]{ON83}.
Thus $J^-(x)\cap \Sigma$ is compact as well.
Let $K\subset M$ be compact with $A\subset J(K)$.
The sets $I^-(x)$ where $x\in M$ form an open cover of $M$.
Hence there are finitely many points $x_1,\ldots,x_n$ such that $K\subset \bigcup_{i=1}^n I^-(x_i)$.
Then we have $J^-(K)\subset \bigcup_{i=1}^n J^-(x_i)$.
Hence $\Sigma\cap J^-(K) \subset \bigcup_{i=1}^n (\Sigma\cap J^-(x_i))$ is compact.

Similarly, one shows that $\Sigma\cap J^+(K)$ is compact.
Thus $\Sigma\cap A\subset \Sigma\cap J(K)$ is compact as well.
\end{proof}

Recall that since all Cauchy hypersurfaces are homeomorphic they are all compact or all noncompact.

\begin{lem}\label{lem:compactCauchy}
The globally hyperbolic manifold $M$ is spacially compact if and only if it has compact Cauchy hypersurfaces.
\end{lem}

\begin{proof}
If the Cauchy hypersurfaces are compact, let $\Sigma$ be one of them.
Then $M=J(\Sigma)$, hence $M$ is spacially compact.

Conversely, if $M$ is spacially compact, then \lref{lem:SpacelikeCompact} with $A=M$ shows that the Cauchy hypersurfaces are compact.
\end{proof}

\begin{lem}\label{lem:pastcompact}
Let $M$ be globally hyperbolic and spacially compact.
Let $A\subset M$ be closed.
Then the following are equivalent:
\begin{enumerate}[(i)]
\item\label{eq:pastcompact1}
$A$ is strictly past compact;
\item\label{eq:pastcompact2}
$A$ is past compact;
\item\label{eq:pastcompact3}
some Cauchy temporal function $t:M\to\R$ attains its minimum on $A$;
\item\label{eq:pastcompact4}
all Cauchy temporal functions $t:M\to\R$ attain their minima on $A$.
\end{enumerate}
\end{lem}

% \begin{proof}
% The implications ``\eqref{eq:pastcompact1}$\Rightarrow$\eqref{eq:pastcompact2}'' and ``\eqref{eq:pastcompact4}$\Rightarrow$\eqref{eq:pastcompact3}'' are clear.
% We show ``\eqref{eq:pastcompact3}$\Rightarrow$\eqref{eq:pastcompact1}''.
% Let $s\in\R$ be the minimum of $t$ on $A$ and let $\Sigma=t^{-1}(s)$ be the corresponding Cauchy hypersurface.
% Then $A\subset J^+(\Sigma)$ and since $\Sigma$ is compact by \lref{lem:compactCauchy}, this shows that $A$ is strictly past compact.
% 
% It remains to show ``\eqref{eq:pastcompact2}$\Rightarrow$\eqref{eq:pastcompact4}''.
% Let $A$ be past compact.
% Let $t:M\to \R$ be a Cauchy temporal function and let $s\in\R$ be in the image of $t$.
% Let $\Sigma=t^{-1}(s)$ be the corresponding Cauchy hypersurface.
% The sets $U_x:=\Sigma\cap I^-(x)$ where $x\in M$ provide an open cover of $\Sigma$.
% Since $\Sigma$ is compact, there exist $x_1,\ldots,x_n$ such that $\Sigma\subset \bigcup_{i=1}^n I^-(x_i)$.
% Then $t^{-1}((-\infty,s])=J^-(\Sigma)\subset \bigcup_{i=1}^n I^-(x_i)$.
% Since each $J^-(x_i)\cap A$ is compact, $t$ attains its minimum on every $J^-(x_i)\cap A$ and hence on $A$. 
% \end{proof}

\begin{proof}
Since the Cauchy hypersurfaces of $M$ are compact, Lemmas~\ref{lem:CauchyChar1} and \ref{lem:CauchyChar2} show ``\eqref{eq:pastcompact1}$\Leftrightarrow$\eqref{eq:pastcompact2}''.
The implication ``\eqref{eq:pastcompact4}$\Rightarrow$\eqref{eq:pastcompact3}'' is clear.
To show ``\eqref{eq:pastcompact1}$\Rightarrow$\eqref{eq:pastcompact4}'' let $A\subset J^+(K)$ for some compact subset $K\subset M$ and let $t$ be a Cauchy temporal function.
Choose $T$ larger than the infimum of $t$ on $A$.
Since $A\cap J^-(t^{-1}(T))$ is contained in the compact set $J^+(K)\cap t^{-1}((-\infty,T) = J^+(K)\cap J^-(t^{-1}(T))$, the function $t$ attains its minimum $t_0$ on this set.
On the rest of $A$, the values of $t$ are even larger than $T$, hence $t_0$ is the minimum of $t$ on all of $A$.

As to ``\eqref{eq:pastcompact3}$\Rightarrow$\eqref{eq:pastcompact2}'', let $t:M\to\R$ be a Cauchy temporal function which attains its minimum on $A$.
By composing with an orientation-preserving diffeomorphism $t(M)\to\R$, we may w.l.o.g.\ assume that $t$ is surjective.
Now \lref{lem:CauchyChar1} shows that $A$ is past compact.
\end{proof}

\begin{rem}\label{rem:Mspacelikecompact}
If $M$ is spacially compact, then every closed subset of $A\subset M$ is spacially compact.
Moreover, if $A$ is temporally compact, then any Cauchy temporal function $t:M\to \R$ attains its maximum $s_+$ and its minimum $s_-$ by \lref{lem:pastcompact}.
Thus $A\subset t^{-1}([s_-,s_+]) \approx \Sigma\times[s_-,s_+]$ where $\Sigma=t^{-1}(s_-)$ is a Cauchy hypersurface.
Since $\Sigma$ is compact, so is $A$.

Summarizing, \diagref{diag:closed1} of implications for closed subsets simplifies as follows for spacially compact $M$:
\begin{diagram}\label{diag:closed2}
\begin{equation*}
\xymatrix{
\mbox{strictly past compact} \ar@{<=>}[d]& \mbox{compact}\ar@{=>}[l]\ar@{=>}[r]\ar@{<=>}[d] & \mbox{strictly future compact}\ar@{<=>}[d] \\
\mbox{past compact} & \mbox{temporally compact}\ar@{=>}[l]\ar@{=>}[r] & \mbox{future compact}
}
\end{equation*}

\vspace{2mm}
\diagramnumber{Closed subsets of a spacially compact manifold}
\end{diagram}
\end{rem}

\subsection{Duality}
We will need the following duality result:

\begin{lem}\label{lem:duality}
Let $M$ be globally hyperbolic and let $A\subset M$ be closed.
Then the following holds:
\begin{enumerate}[(i)]
\item\label{eq:pastcompact}
$A$ is past compact if and only if $A\cap B$ is compact for all strictly future compact sets $B$;
\item\label{eq:futurecompact}
$A$ is future compact if and only if $A\cap B$ is compact for all strictly past compact sets $B$;
\item\label{eq:timelikecompact}
$A$ is temporally compact if and only if $A\cap B$ is compact for all spacially compact sets $B$;
\item\label{eq:strictlypastcompact}
$A$ is strictly past compact if and only if $A\cap B$ is compact for all future compact sets $B$;
\item\label{eq:strictlyfuturecompact}
$A$ is strictly future compact if and only if $A\cap B$ is compact for all past compact sets $B$;
\item\label{eq:spacelikecompact}
$A$ is spacially compact if and only if $A\cap B$ is compact for all temporally compact sets $B$.
\end{enumerate}
\end{lem}

\begin{proof}
(a)
We show \eqref{eq:pastcompact}.
If $A\cap B$ is compact for every strictly future compact $B$, then, in particular, $A\cap J^-(x)$ is compact for every $x\in M$.
Hence $A$ is past compact.

Conversely, let $A$ be past compact and $B$ be strictly future compact.
% We choose a compact subset $K\subset M$ with $B\subset J^-(K)$.
% By compactness of $K$, there exist finitely many points $x_1,\ldots,x_n\in M$ with $K\subset I^-(x_1) \cup \cdots \cup I^-(x_n)$.
% Then $B\subset J^-(K) \subset J^-(x_1) \cup \cdots \cup J^-(x_n)$ and therefore
% \[
% A\cap B \subset (A\cap J^-(x_1)) \cup \cdots \cup (A\cap J^-(x_n)).
% \]
% Since $A$ is past compact, every $A\cap J^-(x_i)$ is compact.
% Hence $A\cap B$ is compact.
Then $A\subset J^+(\Sigma)$ and $B\subset J^-(K)$ for some Cauchy hypersurface $\Sigma\subset M$ and some compact subset $K\subset M$.
Thus $A\cap B \subset J^+(\Sigma) \cap J^-(K)$, hence $A\cap B$ is contained in a compact set, hence compact itself.

(b)
The proof of \eqref{eq:futurecompact} is analogous.
As to \eqref{eq:timelikecompact}, if $A\cap B$ is compact for every spacially compact $B$, then, in particular, $A\cap J^+(x)$ and $A\cap J^-(x)$ are compact for every $x\in M$.
Hence $A$ is temporally compact.

Conversely, let $A$ be temporally compact and $B$ be spacially compact.
We choose a compact $K\subset M$ with $B\subset J(K)$.
By \eqref{eq:pastcompact}, $A\cap J^-(K)$ is compact and by \eqref{eq:futurecompact}, $A\cap J^+(K)$ is compact.
Thus $A\cap B \subset A\cap J(K) = (A\cap J^+(K)) \cup (A\cap J^-(K))$ is compact.

(c)
We show \eqref{eq:strictlypastcompact}.
By \eqref{eq:futurecompact} the intersection of a strictly past compact set and a future compact set is compact.
Now assume $A$ is not strictly past compact.
We have to find a future compact set $B$ such that $A\cap B$ is noncompact.
Let $K_1\subset K_2 \subset K_3 \subset \cdots \subset M$ be an exhaustion by compact subsets.
We choose the exhaustion such that every compact subset of $M$ is contained in $K_j$ for sufficiently large $j$.
Since $A$ is not strictly past compact there exists $x_j\in A\setminus J^+(K_j)$ for every $j$.
The set $B:=\{x_1,x_2,x_3,\ldots\}$ is not compact because otherwise, for sufficiently large $j$, we would have $B\subset K_j \subset J^+(K_j)$ contradicting the choice of the $x_i$.
But $B$ is future compact.
Namely, let $x\in M$.
Then $x\in K_j$ for $j$ large and therefore $B\cap J^+(x) \subset B\cap J^+(K_j) \subset \{x_1,\ldots,x_{j-1}\}$ is finite, hence compact.
Now $A\cap B=B$ is not compact which is what we wanted to show.

(d)
The proof of \eqref{eq:strictlyfuturecompact} is analogous.
As to \eqref{eq:spacelikecompact}, we know already by \eqref{eq:timelikecompact} that the intersection of a temporally compact and a spacially compact set is always compact.
If $A$ is not spacially compact, then the same construction as in the proof of \eqref{eq:strictlypastcompact} with $J^+(K_j)$ replaced by $J(K_j)$ yields a noncompact set $B\subset A$ which is temporally compact.
This concludes the proof.
\end{proof}

\subsection{Causal compatibility}
An open subset $N$ of a time oriented Lorentzian manifold $M$ is a time oriented Lorentzian manifold itself.
We call $N$ \emph{causally compatible} if $J^\pm_N(x)=J^\pm_M(x)\cap N$ for all $x\in N$.
In other words, any two points in $N$ which can be connected by causal curve in $M$ can also be connected by causal curve that stays in $N$.

%%%%%%%%%%%%%%%%%%%%%%%%%%%%%%%%%%%%%%%%%%%%%%%%%%%%%%%%%%%%%%%%%%%%%%%%%%%%%%%%%%%%%%%%%%%%%%%%%

\section{The Function Spaces}

Throughout this section, let $M$ denote a globally hyperbolic Lorentzian manifold.
In particular, $M$ carries a time-orientation and an induced volume element which we denote by $\dV$.
Moreover, let $E\to M$ be a (real or complex, finite dimensional) vector bundle.

\subsection{Smooth sections}
We denote the space of smooth sections of $E$ by $\Coo(M,E)$.
Any connection $\na$ on $E$ induces, together with the Levi-Civita connection on $T^*M$, a connection on $T^*M^{\otimes \ell}\otimes E$ for any $\ell\in\N$.
For any $f\in\Coo(M,E)$, the $\ell^\mathrm{th}$ covariant derivative $\na^{\ell}f := \na\cdots\na\na f$ is a smooth section of  $T^*M^{\otimes \ell}\otimes E$.

For any compact subset $K\subset M$, any $m\in\N$, any connection $\na$ on $E$ and any auxiliary norms $|\cdot |$ on $T^*M^{\otimes \ell}\otimes E$ we define the semi-norm
\[
\|f\|_{K,m,\na,|\cdot |} := \max_{\ell=0,\cdots,m}\,\,\max_{x\in K} |\na^{\ell}f(x)|
\]
for $f\in\Coo(M,E)$.
By compactness of $K$, different choices of $\na$ and $|\cdot |$ lead to equivalent semi-norms.
For this reason, we may suppress $\na$ and $|\cdot |$ in the notation and write $\|f\|_{K,m}$ instead of $\|f\|_{K,m,\na,|\cdot |}$.
This family of semi-norms is separating and turns $\Coo(M,E)$ into a locally convex topological vector space.
If we choose a sequence $K_1 \subset K_2 \subset K_3 \subset \cdots\subset M$ of compact subsets with $\bigcup_{i=1}^\infty K_i =M$ and such that each $K_i$ is contained in the interior of $K_{i+1}$, then the countable subfamily $\|\cdot\|_{K_i,i}$ of semi-norms is equivalent to the original family.
Hence $\Coo(M,E)$ is metrizable.
An Arzel\`a-Ascoli argument shows that  $\Coo(M,E)$ is complete.
Thus  $\Coo(M,E)$ is a Fr\'echet space.
A sequence of sections converges in $\Coo(M,E)$ if and only if the sections and all their (higher) derivatives converge locally uniformly.

\subsection{Support systems}
For a closed subset $A\subset M$ denote by $\CooA(M,E)$ the space of all smooth sections $f$ of $E$ with $\supp f\subset A$.
Then $\CooA(M,E)$ is a closed subspace of $\Coo(M,E)$ and hence a Fr\'echet space in its own right.
Moreover, if $A_1\subset A_2$ then $\Coo_{A_1}(M,E)$ is a closed subspace of $\Coo_{A_2}(M,E)$.

We denote by $\CM$ the set of all closed subsets of $M$.

\begin{dfn}\label{def:suppsys}
A subset $\A\subset\CM$ is called a \emph{support system} on $M$ if the following holds:
\begin{enumerate}[(i)]
\item\label{suppsys1}
For any $A,A'\in\A$ we have $A\cup A'\in \A$;
\item\label{suppsys2}
For any $A\in\A$ there is an $A'\in\A$ such that $A$ is contained in the interior of $A'$;
\item\label{suppsys3}
If $A\in\A$ and $A'\subset A$ is a closed subset, then $A'\in\A$.
\end{enumerate}
\end{dfn}

The first condition implies that $\A$ is a direct system with respect to inclusion.
The third condition is harmless;
if $\A$ satisfies \eqref{suppsys1} and \eqref{suppsys2}, then adding all closed subsets of the members of $\A$ to $\A$ will give a support system.

Given a support system on $M$ we obtain the direct system $\{\CooA(M,E)\}_{A\in\A}$ of subspaces of $\Coo(M,E)$ and denote by $\Coo_\A(M,E)$ its direct limit as a locally convex topological vector space.
As a vector subspace, $\Coo_\A(M,E)$ is simply $\bigcup_{A\in\A}\CooA(M,E)$.
A convex subset $\OO\subset\Coo_\A(M,E)$ is open if and only if $\OO\cap \CooA(M,E)$ is open for all $A\in\A$.
Note that $\Coo_\A(M,E)$ is not a closed subspace of $\Coo(M,E)$ in general.

\begin{dfn}
We call a support system \emph{essentially countable} if
there is a sequence $A_1, A_2, A_3,\ldots\in\A$ such that each $A_j\subset A_{j+1}$ and for any $A\in\A$ there exists a $j$ with $A\subset A_j$.
Such a sequence $A_1\subset A_2 \subset A_3 \subset \cdots$ is called a \emph{basic chain} of $\A$.
\end{dfn}

\begin{lem}\label{lem:LFKonvergenz}
Let $\A\subset\CM$ be an essentially countable support system on $M$.
If $V\subset\Coo_\A(M,E)$ is a bounded subset then there exists an $A\in\A$ such that $V\subset\CooA(M,E)$.
In particular, for any convergent sequence $f_j\in\Coo_\A(M,E)$ there exists an $A\in\A$ such that $f_j\in\CooA(M,E)$ for all $j$.
\end{lem}

This shows that a sequence $(f_j)$ converges in $\Coo_\A(M,E)$ if and only if there exists an $A\in\A$ such that $f_j\in\CooA(M,E)$ for all $j$ and $(f_j)$ converges in $\CooA(M,E)$.

\begin{proof}[Proof of Lemma~\ref{lem:LFKonvergenz}]
Consider a basic chain $A_1\subset A_2\subset A_3\subset\ldots$.
Let $V\subset\Coo_\A(M,E)$ be a subset not contained in any $\Coo_{A_j}(M,E)$.
We have to show that $V$ is not bounded.
Pick points $x_j\in M\setminus A_j$ and sections $f_j\in V$ with $f_j(x_j)\neq 0$.
Define the convex set
\[
W:=\bigg\{f\in\Coo_\A(M,E) \,\bigg|\, |f(x_j)|<\frac{|f_j(x_j)|}{j} \mbox{ for all }j\bigg\}. 
\]
Each $A\in\A$ contains only finitely many $x_j$.
Thus $W\cap\CooA(M,E)=\{f\in\CooA(M,E)\mid \|f\|_{\{x_j\},0}<|f_j(x_j)|/j\}$ is open in $\CooA(M,E)$.
Therefore $W$ is an open neighborhood of $0$ in $\Coo_\A(M,E)$.

For any $T>0$ we have $T\cdot W=\{f\in\Coo_\A(M,E) \mid |f(x_j)|<\frac{T}{j}|f_j(x_j)| \mbox{ for all }j\}$ and hence $f_j\notin TW$ for $j>T$.
Thus $V$ is not contained in any $TW$ and is therefore not bounded.
\end{proof}

\begin{ex}
The system $\A=\CM$ of all closed subsets is an essentially countable support system on $M$.
A basic chain is given by the constant sequence $M\subset M \subset M \subset \cdots$.
Clearly, $\Coo_{\CM}(M,E)=\Coo(M,E)$.
\end{ex}

\begin{ex}
Let $\A=c$ where $c$ is the set of all compact subsets of $M$.
A basic chain can be constructed as follows:
Provide $M$ with a complete Riemannian metric $\gamma$.
Fix a point $x\in M$.
Now let $A_j$ be the closed ball centered at $x$ with radius $j$ with respect to $\gamma$.

Then $\Cooc(M,E)$ is the space of compactly supported smooth sections, also called \emph{test sections}.
\end{ex}

\begin{ex}
Let $\A=sc$ be the set of all spacially compact subsets of $M$.
If $K_1\subset K_2 \subset K_3\subset \cdots$ is a basic chain of $c$, then $J(K_1)\subset J(K_2) \subset J(K_3)\subset \cdots$ is a basic chain of $sc$.
Hence $sc$ is essentially countable.

Now $\Coosc(M,E)$ is the space of smooth sections with spacially compact support.
Recall that a sequence $(f_j)$ converges in $\Coosc(M,E)$ if and only if there exists a compact subset $K\subset M$ such that $\supp(f_j)\subset J(K)$ for all $j$ and $(f_j)$ converges locally uniformly with all derivatives.
\end{ex}

\begin{ex}
Let $\A=spc$ be the set of all strictly past compact subsets of $M$.
As in the previous example we see that $spc$ is essentially countable.
Now $\Coospc(M,E)$ is the space of smooth sections with strictly past-compact support.

Similarly, one can define the space $\Coosfc(M,E)$ of smooth sections with strictly future-compact support.
\end{ex}

\begin{ex}
Let $\A=pc$ be the set of all past-compact subsets.
If $M$ is spacially compact then $pc=spc$ by \lref{lem:pastcompact} but in general $pc$ is strictly larger than $spc$.
We obtain the space $\Coopc(M,E)$ of smooth sections with past-compact support.

In general, the support system $pc$ is not essentially countable.
The following example was communicated to me by Miguel S\'anchez.
Let $M$ be the $(1+1)$-dimensional Minkowski space.
Let $A_1\subset A_2\subset A_3\subset \cdots \subset M$ be a chain of past-compact subsets.
Look at the ``future-diverging'' sequence of points $(n,0)\in M$ and choose points\footnote{Note that $M\setminus I^-(n,0)$ is not past compact so that $A_n\cup J^-(n,0)$ cannot be all of $M$. Compare with \eref{ex:pastcompactbutnotstrict} however.} $p_n\in M\setminus (A_n \cup J^-(n,0))$.
By construction, $A:=\{p_1,p_2,p_3,\ldots\}$ is not contained in any $A_n$ but $A$ is past compact.
Namely, let $x\in M$.
Then there exists an $n$ such that $x\in J^-(n,0)$.
Now $J^-(x)\cap A \subset J^-(n,0)\cap A$ is finite and hence compact.
Thus no chain in $pc$ captures all elements of $pc$, so $pc$ is not essentially countable.
\end{ex}

\begin{ex}
A similar discussion as in the previous example yields the space $\Coofc(M,E)$ of smooth sections with future-compact support and the space $\Cootc(M,E)$ of smooth sections with temporally compact support.
Both support systems are not essentially countable in general.
But again, if $M$ is spacially compact, they are because then $fc=sfc$ and $tc=c$ by \rref{rem:Mspacelikecompact}.
\end{ex}

If $\A\subset\A'$, then $\Coo_\A(M,E) \subset \Coo_{\A'}(M,E)$ and the inclusion map is continuous.
Hence we obtain the following diagram of continuously embedded spaces:

\begin{diagram}\label{diag:Cspaces}
\begin{equation*}
\xymatrix@R=8pt{
           & \Coospc(M,E)\ar@{^{(}->}[rr] \ar@{_{(}->}[ddrr]&& \Coopc(M,E)  \ar@{_{(}->}[ddr]& \\
&&&\\
\Cooc(M,E)\ar@{^{(}->}[uur] \ar@{_{(}->}[ddr] \ar@{^{(}->}[r]& \Cootc(M,E) \ar@{^{(}->}'[ur][uurr]   \ar@{_{(}->}'[dr][ddrr] &&\Coosc(M,E)\ar@{^{(}->}[r]  & \Coo(M,E) \\
&&&\\
           & \Coosfc(M,E) \ar@{^{(}->}[rr] \ar@{^{(}->}[uurr]&& \Coofc(M,E)\ar@{^{(}->}[uur] & 
}
\end{equation*}
\end{diagram}
\diagramnumber{Smooth sections with various support properties}
All embeddings in \diagref{diag:Cspaces} have dense image.
Namely, we have

\begin{lem}\label{lem:Ccdicht1}
Let $\A$ be a support system on $M$ such that $c\subset \A$, i.e., each compact set is contained in $\A$.
Then $\Cooc(M,E)$ is a dense subspace of $\Coo_\A(M,E)$.
\end{lem}

\begin{proof}
Let $f\in\Coo_\A(M,E)$ and let $\OO$ be a convex open neighborhood of $f$ in $\Coo_\A(M,E)$.
Let $A\in\A$ with $f\in\CooA(M,E)$.
Since $\OO\cap\CooA(M,E)$ is open in $\CooA(M,E)$ there exists an $\eps>0$ and a seminorm $\|\cdot\|_{K,m}$ such that
\[
\{g\in\CooA(M,E) \mid \|f-g\|_{K,m}<\eps\}
\subset
\OO\cap\CooA(M,E).
\]
Pick a cutoff function $\chi\in\Cooc(M,\R)$ with $\chi\equiv 1$ on $K$.
Then $g:=\chi\cdot f\in\Cooc(M,E)$ and $\|f-g\|_{K,m}=0$.
Thus $g\in\OO\cap\CooA(M,E)$.
\end{proof}

\subsection{Distributional sections}
Now denote the dual bundle of $E\to M$ by $E^*\to M$.
The canonical pairing \mbox{$E^*\otimes E\to \R$} is denoted by $\<\cdot,\cdot\>$.
A locally integrable section $f$ of $E$ can be considered as a continuous linear functional on $\Cooc(M,E^*)$ by $f[\phi]=\int_M\<\phi,f\>\dV$.
We denote the topological dual space of $\Cooc(M,E^*)$ by $\D(M,E)$.
The elements of $\D(M,E)$, i.e., the continuous linear functionals on $\Cooc(M,E^*)$, are called \emph{distributional sections} of $E$.
It is well known that a distributional section of $E$ has compact support if and only if it extends to a continuous linear functional on $\Coo(M,E^*)$.
We denote the space of distributional sections of $E$ with compact support by $\Dc(M,E)$.

More generally, for any closed subset $A\subset M$ we denote by $\D_A(M,E)$ the space of all distributional sections of $E$ whose support is contained in $A$.
Likewise, for any support system $\A$ on $M$, we denote by $\D_\A(M,E)$ the space of all distributional sections of $E$ whose support is an element of $\A$.
Again, we have the continuous embedding $\Coo_\A(M,E) \hookrightarrow \D_\A(M,E)$ defined by $f[\phi]=\int_M\<\phi,f\>\dV$ for any $f\in\Coo_\A(M,E)$ and any test section $\phi\in\Cooc(M,E^*)$.

\subsection{Duality}
We now characterize the topological dual spaces of the other spaces in \diagref{diag:Cspaces} (with $E$ replaced by~$E^*$).
We will show that the support of a distributional section is contained in a support system if and only if it extends to test sections having their support in a dual support system.

\begin{dfn}\label{def:dual}
Two support systems $\A$ and $\B$ be on $M$ are said to be \emph{in duality} if for any $C\in\CM$:
\begin{enumerate}[(i)]
\item\label{dual1}
$C\in\A$ if and only if $C\cap B$ is compact for all $B\in\B$;
\item\label{dual2}
$C\in\B$ if and only if $C\cap A$ is compact for all $A\in\A$.
\end{enumerate}
\end{dfn}

\begin{ex}
Here are some examples of support systems $\A$ and $\B$ in duality.
The last column contains a justification of this fact.
\begin{center}
\begin{tabular}{|c|c|c|}
\hline
$\A$ & $\B$ & why?\\ \hline
$\CM$ & $c$ & obvious\\
$pc$ & $sfc$ & \lref{lem:duality}~\eqref{eq:pastcompact} and \eqref{eq:strictlyfuturecompact}\\
$fc$ & $spc$ & \lref{lem:duality}~\eqref{eq:futurecompact} and \eqref{eq:strictlypastcompact}\\
$tc$ & $sc$ & \lref{lem:duality}~\eqref{eq:timelikecompact} and \eqref{eq:spacelikecompact}\\
\hline
\end{tabular}

\vspace{5pt}
{\sc Table~1:}
\emph{Support systems in duality}
\end{center}
\end{ex}

\begin{lem}\label{lem:DistriSupp}
Let $\A$ and $\B$ be two support systems on $M$ in duality.
Then a distributional section $f\in\D(M,E)$ has support contained in $\A$ if and only if $f$ extends to a continuous linear functional on $\Coo_\B(M,E^*)$.
\end{lem}

\begin{proof}
(a)
Suppose first that $\supp f\in\A$.
Let $B\in\B$.
Since $\supp f\cap B$ is compact there is a cutoff function $\chi\in\Coo_c(M,\R)$ with $\chi\equiv 1$ on a neighborhood of $\supp f\cap B$.
We extend $f$ to a linear functional $F_B$ on $\Coo_B(M,E^*)$ by 
\[
F_B[\phi] := f[\chi\phi].
\]
This extension is independent of the choice of $\chi$ because for another choice $\chi'$, $f$ and $\chi\phi-\chi'\phi$ have disjoint supports.
If $\phi_j\to0$ in $\Coo_B(M,E^*)$, then $\chi\phi_j\to0$ in $\Coo_c(M,E^*)$ and hence $F_B[\phi_j] = f[\chi\phi_j]\to0$.
Thus $F_B$ is continuous.

Doing this for every $B\in\B$ we obtain an extension $F$ of $f$ to a linear functional on $\Coo_\B(M,E^*)$ with $F_B$ being the restriction of $F$ to $\Coo_B(M,E^*)$.
Continuity of $F$ holds because each $F_B$ is continuous.

(b)
Conversely, assume that $f$ extends to a continuous linear functional $F$ on $\Coo_\B(M,E^*)$.
We check that $\supp f\in\A$ by showing that $\supp f\cap B$ is compact for every $B\in\B$.

Let $B\in\B$.
Choose $B'\in\B$ such that $B$ is contained in the interior of $B'$.
Since the restriction $F_{B'}$ of $F$ to $\Coo_{B'}(M,E^*)$ is linear and continuous, there exists a seminorm $\|\cdot\|_{K,m}$ and a constant $C>0$ such that 
\[
|F_{B'}[\phi]| \leq C\cdot \|\phi\|_{K,m}
\]
for all $\phi\in\Coo_{B'}(M,E^*)$.
In particular, $F_{B'}[\phi]=0$ if $\supp(\phi)$ and $K$ are disjoint.

\emph{Claim:}
$B\cap (M\setminus K) \subset M\setminus \supp(F)$.

Namely, let $x\in B\cap (M\setminus K)$.
Then $x$ lies in the interior of $B'$.
Hence there is an open neighborhood $U$ of $x$ entirely contained in $B'$.
Since $x\notin K$ we may assume that $U$ and $K$ are disjoint.
Now we know that for all $\phi\in\Cooc(M,E^*)$ with $\supp(\phi)\subset U$ we have $F[\phi]=0$.
Thus $x\notin\supp(F)$.
\hfill\ding{51}

The claim implies $\supp(F)\subset (M\setminus B)\cup K$ and hence $\supp(F)\cap B \subset K$.
Therefore the intersection $\supp(F)\cap B$ is compact.
\end{proof}

\begin{rem}
Observe that for the proof of \lref{lem:DistriSupp} we only need \eqref{dual1} in \dref{def:dual} but not \eqref{dual2}.
\end{rem}

Dualizing \diagref{diag:Cspaces}, Table~1 and \lref{lem:DistriSupp} yield the following diagram of continuous embeddings of several spaces of distributions, characterized by different support properties:
\begin{diagram}\label{diag:Dspaces}
\begin{equation*}
\xymatrix@R=8pt{
           & \Dfc(M,E)\ar@{^{(}->}[ddl] && \Dsfc(M,E) \ar@{_{(}->}[ll]\ar@{^{(}->}'[dl][ddll] & \\
&&&\\
\D(M,E)& \Dsc(M,E)  \ar@{_{(}->}[l]  &&\Dtc(M,E)\ar@{_{(}->}[uull]\ar@{^{(}->}[ddll] & \Dc(M,E) \ar@{_{(}->}[l]  \ar@{_{(}->}[uul] \ar@{^{(}->}[ddl]  \\
&&&\\
           & \Dpc(M,E)\ar@{_{(}->}[uul] && \Dspc(M,E)\ar@{_{(}->}[ll]\ar@{_{(}->}'[ul][uull] & 
}
\end{equation*}
\end{diagram}
\diagramnumber{Distributional sections with various support properties}

\subsection{Convergence of distributions}
Let $\A$ be one of the support systems $\CC$, $pc$, $fc$, $tc$, $sc$, $spc$, $sfc$, or $c$.
Let $\B$ be the dual support system as in Table~1.
The continuity property of a distributional section $f\in\D_\A(M,E)$ means that for any $B\in\B$ the restriction of $f$ to $\Coo_B(M,E^*)$ is continuous.
In other words, for any sequence of smooth sections $\phi_j$ with support contained in $B$ which converge locally uniformly with their derivatives to some $\phi\in\Coo_B(M,E^*)$, we must have $f[\phi_j]\to f[\phi]$.

Our distribution spaces are always equipped with the weak*-topology.
This means that a sequence $f_j\in\D_\A(M,E)$ converges if and only if $f_j[\phi]$ converges for every fixed $\phi\in\Coo_\B(M,E^*)$.

We have the following analog to \lref{lem:Ccdicht1}:

\begin{lem}\label{lem:Ccdicht2}
Let $\A$ be one of the support systems $\CC$, $pc$, $fc$, $tc$, $sc$, $spc$, $sfc$, or $c$.
Then $\Cooc(M,E)$ is a dense subspace of $\D_\A(M,E)$.
\end{lem}

\begin{proof}
Let $\B$ be the dual support system to $\A$ as in Table~1.
Let $u\in\D_\A(M,E)$.
Put $A:=\supp(u)$, hence $u\in\D_A(M,E)$.
It is well known that $\Cooc(M,E)$ is dense in $\D(M,E)$.
Hence there is a sequence $u_j\in\Cooc(M,E)$ with $u_j\to u$ in $\D(M,E)$.

Choose $A'\in\A$ such that $A$ is contained in the interior of $A'$.
Let $\chi\in\Coo(M,\R)$ be a function such that $\chi\equiv 1$ on $A$ and $\supp\chi\subset A'$.

Let $\phi\in\Coo_B(M,E^*)$ where $B\in\B$.
Since $A'\cap B$ is compact, the section $\chi\phi$ has compact support.
Therefore
\[
(\chi u_j)[\phi] = u_j[\chi\phi] \to u[\chi\phi] = (\chi u)[\phi] = u[\phi].
\]
Thus the compactly supported sections $\chi u_j$ converge to $u$ in $\D_\A(M,E)$.
\end{proof}

%%%%%%%%%%%%%%%%%%%%%%%%%%%%%%%%%%%%%%%%%%%%%%%%%%%%%%%%%%%%%%%%%%%%%%%%%%%%%%%%%%%%%%%%%%%%%%%%%

\section{Properties of Green-Hyperbolic Operators}

\subsection{Green's operators and Green hyperbolic operators}
Let $E_1,E_2\to M$ be vector bundles over a globally hyperbolic manifold.
Let $P:\Coo(M,E_1) \to \Coo(M,E_2)$ be a linear differential operator.
Differential operators do not increase supports and yield continuous maps $P:\Coo_\A(M,E_1) \to \Coo_\A(M,E_2)$ for any support system $\A$.

There is a unique linear differential operator $\Pt:\Coo(M,E_2^*) \to \Coo(M,E_1^*)$ characterized by
\begin{equation}
\int_M \<\phi,Pf\> \dV = \int_M \<\Pt\phi,f\> \dV
\label{eq:FormalDual}
\end{equation}
for all $f\in\Coo(M,E_1)$ and $\phi\in\Coo(M,E_2^*)$ such that $\supp f \cap \supp(\phi)$ is compact.
Here again, $\<\cdot,\cdot\>$ denotes the canonical pairing of $E_i^*$ and $E_i$.
The operator $\Pt$ is called the \emph{formally dual operator} of $P$.

\begin{dfn}\label{def:GreenOp}
An \emph{advanced Green's operator} of $P$ is a linear map $G_+:\Cooc(M,E_2)\to\Coo(M,E_1)$ such that
\begin{enumerate}[(i)]
\item\label{eq:GreenOp1}
$G_+Pf=f$ for all $f\in\Cooc(M,E_1)$;
\item\label{eq:GreenOp2}
$PG_+f=f$ for all $f\in\Cooc(M,E_2)$;
\item\label{eq:GreenOp3}
$\supp(G_+f)\subset J^+(\supp f)$ for all $f\in\Cooc(M,E_2)$.
\end{enumerate}
A linear map $G_-:\Cooc(M,E_2)\to\Coo(M,E_1)$ is called a \emph{retarded Green's operator} of $P$ if \eqref{eq:GreenOp1}, \eqref{eq:GreenOp2} hold and 
\begin{enumerate}[(i)']
\setcounter{enumi}{2}
\item \label{eq:GreenOp4}
$\supp(G_-f)\subset J^-(\supp f)$ holds for every $f\in\Cooc(M,E_2)$.
\end{enumerate}
\end{dfn}
 
\begin{dfn}
The operator $P$ is be called \emph{Green hyperbolic} if $P$ and $\Pt$ have advanced and retarded Green's operators.
\end{dfn}

We will see in \cref{cor:GreenUnique} that uniqueness of the Green's operators comes for free.

\begin{ex}\label{ex:waveop}
The most prominent examples of Green-hyperbolic operators are \emph{wave operators}, also called \emph{normally hyperbolic operators}.
They are second-order differential operators $P$ whose principal symbol is given by the Lorentzian metric.
Locally they take the form
\[
P = \sum_{ij} g^{ij}(x)\frac{\partial^2}{\partial x^i \partial x^j} + \sum_j B_j(x) \frac{\partial}{\partial x^j} + C(x)
\]
where $g^{ij}$ denote the components of the inverse metric tensor, and $B_j$ and $C$ are matrix-valued coefficients depending smoothly on $x$.

The class of wave operators contains the \emph{d'Alembert operator} $P=\Box$, the \emph{Klein-Gordon operator} $P=\Box + m^2$, and the Klein-Gordon operator with a potential, $P=\Box + V$.
In these cases, the operator acts on functions, i.e., the underlying vector bundles $E_1$ and $E_2$ are simply trivial line bundles.

On any vector bundle $E$ one may choose a connection $\nabla$ and put $P=\mathrm{tr}(\nabla^2)$ to obtain a wave operator $\Coo(M,E)\to\Coo(M,E)$.
If $E=\Lambda^kT^*M$ is the bundle of $k$-forms, then $P=d\delta + \delta d$ is a wave operator where $d$ denotes the exterior differential and $\delta$ the codifferential.

It is shown in \cite[Cor.~3.4.3]{BGP07} that wave operators have Green's operators.
Since the formally dual operator of a wave operator is again a wave operator, wave operators are Green hyperbolic.
\end{ex}

\begin{ex}\label{ex:Mink2}
Let us consider a concrete special case of \eref{ex:waveop}.
Let $M=\R^2$ be $2$-dimensional Minkowski space.
We denote a generic point of $M$ by $(t,x)$.
Let 
\[P=\Box=-\frac{\partial^2}{\partial t^2}+\frac{\partial^2}{\partial x^2}\] 
be the d'Alembert operator.
Then one checks by explicit calculation that 
\[
(G_+f)(t,x) 
=
-\frac12 \int_{J^-(t,x)} f(\tau,\xi) \,d\xi \, d\tau
=
-\frac12 \int_{-\infty}^t\left(\int_{x+\tau -t}^{x+t-\tau} f(\tau,\xi) \, d\xi\right) d\tau
\]
yields an advanced Green's operator for $\Box$.
Replacing $J^-(t,x)$ by $J^+(t,x)$ we get a retarded Green's operator.
In other words, the integral kernel of $G_+$ is $-\frac12$ times the characteristic function of $\{(t,x,\tau,\xi)\mid (\tau,\xi)\in J^-(t,x)\}=\{(t,x,\tau,\xi)\mid(x-\xi)^2\leq (t-\tau)^2, \, \tau\leq t\}\subset M\times M$.

For the d'Alembert operator on higher-dimensional Minkowski space the integral kernel of $G_\pm$ is no longer an $L^\infty$-function but is given by the so-called Riesz distributions, see \cite[Sec.~1.2]{BGP07}.
\end{ex}

\begin{ex}
Let $E=T^*M$ and $m>0$.
Then $P=\delta d + m^2$ is the \emph{Proca operator}.
Now $\tilde P:=d\delta + \delta d + m^2$ is a wave operator and hence has Green's operators $\tilde G_\pm$.
One can check that $G_\pm := (m^{-2}d\delta + \mathrm{id})\circ \tilde G_\pm$ are Green's operators of $P$, compare \cite[Sec.~2.4]{BG12b}.
Similarly, one gets Green's operators for $\Pt$.
Thus the Proca operator is not a wave operator but it is Green hyperbolic.
\end{ex}

\subsection{Restrictions to subregions}
Green hyperbolicity persists under restriction to suitable subregions of the manifold $M$.

\begin{lem}
Let $M$ be globally hyperbolic and let $N\subset M$ be an open subset which is causally compatible and globally hyperbolic.
Then the restriction of $P$ to $N$ is again Green hyperbolic.
\end{lem}

\begin{proof}
We construct an advanced Green's operator for the restriction $P|_N$ of $P$ to $N$.
The construction of the retarded Green's operator and the ones for $\Pt$ are analogous.
Denote by $\ext:\Cooc(N,E_2|_N)\to\Cooc(M,E_2)$ the extension-by-zero operator and by $\res:\Coo(M,E_1)\to\Coo(N,E_1|_N)$ the restriction operator.
Let $G_+:\Cooc(M,E_2) \to \Coo(M,E_1)$ be the advanced Green's operator of $P$.
We claim that 
\[
G_+^N:=\res\circ G_+\circ\ext:\,\,\,\Cooc(N,E_2|_N)\to\Coo(N,E_1|_N)
\]
is an advanced Green's operator of $P|_N$.
Since differential operators commute with restrictions and extensions we easily check for $f\in\Cooc(N,E_i|_N)$:
\[
P|_N(G_+^Nf)
=
\res\circ P\circ G_+ \circ \ext f
=
\res\circ\ext f
=
f
\]
and
\[
G_+^N(P|_Nf)
=
\res \circ G_+\circ \ext \circ \res \circ P \circ \ext f
=
\res \circ G_+\circ  P \circ \ext f
=
\res \circ \ext f
=
f.
\]
This shows \eqref{eq:GreenOp1} and \eqref{eq:GreenOp2} in \dref{def:GreenOp}.
As to \eqref{eq:GreenOp3} we see
\begin{align*}
\supp(G_+^Nf)
&=
\supp(\res\circ G_+\circ\ext f)
=
\supp(G_+\circ\ext f) \cap N \\
&\subset
J^+_M(\supp(\ext f)) \cap N
=
J^+_M(\supp f) \cap N
=
J^+_N(\supp f).
\end{align*}
In the last equality we used that $N$ is causally compatible.
\end{proof}

\begin{dfn}
Let $G_\pm$ be advanced and retarded Green's operators of $P$.
Then $G:=G_+-G_-:\Cooc(M,E_2) \to \Coo(M,E_1)$ is called the \emph{causal propagator}.
\end{dfn}

\subsection{Extensions of Green's operators}
From \eqref{eq:GreenOp3} and \eqref{eq:GreenOp3}' in \dref{def:GreenOp} we see that the Green's operators of $P$ give rise to linear maps
\begin{align*}
G_+&:\Cooc(M,E_2) \to \Coospc(M,E_1), \\
G_-&:\Cooc(M,E_2) \to \Coosfc(M,E_1), \\
G&:\Cooc(M,E_2) \to \Coosc(M,E_1).
\end{align*}

\begin{thm}\label{thm:Gquer}
There are unique linear extensions
\begin{align*}
\overline G_+:\Coopc(M,E_2) \to \Coopc(M,E_1) \quad\mbox{ and }\quad
\overline G_-:\Coofc(M,E_2) \to \Coofc(M,E_1)
\end{align*}
of $G_+$ and $G_-$ respectively, such that 
\begin{enumerate}[(i)]
\item\label{eq:GreenOp5}
$\overline G_+Pf=f$ for all $f\in\Coopc(M,E_1)$;
\item\label{eq:GreenOp6}
$P\overline G_+f=f$ for all $f\in\Coopc(M,E_2)$;
\item\label{eq:GreenOp7}
$\supp(\overline G_+f)\subset J^+(\supp f)$ for all $f\in\Coopc(M,E_2)$;
\end{enumerate}
and similarly for $\overline G_-$.
\end{thm}

\begin{proof}
We only consider $\overline G_+$, the proof for $\overline G_-$ being analogous.

(a)
Let $f\in\Coopc(M.E)$.
Given $x\in M$ we define $(\overline G_+f)(x)$ as follows:
Since $J^-(x)\cap\supp f$ is compact we can choose a cutoff function $\chi\in\Cooc(M,\R)$ with $\chi\equiv 1$ on a neighborhood of $J^-(x)\cap\supp f$.
Now we put 
\begin{equation}
(\overline G_+f)(x):=(G_+(\chi f))(x).
\label{eq:defGquer}
\end{equation}

(b)
The definition in \eqref{eq:defGquer} is independent of the choice of $\chi$.
Namely, let $\chi'$ be another such cutoff function.
It suffices to show $x\notin \supp(G_+((\chi-\chi')f))$.
If $x\in \supp(G_+((\chi-\chi')f))\subset J^+(\supp((\chi-\chi')f))$ then there would be a causal curve from $\supp((\chi-\chi')f)$ to $x$.
Hence $\supp((\chi-\chi')f)\cap J^-(x)$ would be nonempty.
On the other hand,
\begin{align*}
\supp((\chi-\chi')f)\cap J^-(x)
&=
\supp(\chi-\chi')  \cap \supp f \cap J^-(x)\\
&\subset
\supp(\chi-\chi') \cap \{\chi\equiv\chi'\equiv 1\}\\
&=
\emptyset ,
\end{align*}
a contradiction.

(c)
The section $\overline G_+f$ is smooth.
Namely, a cutoff function $\chi$ for $x\in M$ also works for all $x'\in J^-(x)$ simply because $J^-(x')\subset J^-(x)$.
In particular, on the open set $I^-(x)$ we have $\overline G_+f=G_+(\chi f)$ for a fixed $\chi$.
Hence $\overline G_+f$ is smooth on $I^-(x)$.
Since any point in $M$ is contained in $I^-(x)$ for some $x$,  $\overline G_+f$ is smooth on $M$.

(d)
The operator $\overline G_+$ is linear.
The only issue here is additivity.
Let $f_1,f_2\in\Coopc(M,E_2)$.
Then $\supp(f_1) \cap J^-(x)$ and  $\supp(f_2) \cap J^-(x)$ are both compact and we may choose the cutoff function  $\chi$ such that $\chi\equiv 1$ on neighborhoods of both $\supp(f_1) \cap J^-(x)$ and  $\supp(f_2) \cap J^-(x)$.
Then $\chi\equiv 1$ on a neighborhood of $\supp(f_1+f_2) \cap J^-(x)$ and we get
\begin{align*}
(\overline G_+(f_1+f_2))(x)
&=
(G_+(\chi f_1 + \chi f_2))(x) \\
&=
(G_+(\chi f_1)(x) + (G_+(\chi f_2))(x) \\
&=
(\overline G_+ f_1)(x) + (\overline G_+f_2))(x).
\end{align*}

(e)
Let $x\in M$ and $\chi$ a cutoff function which is identically $\equiv 1$ on a neighborhood of $\supp f\cap J^-(x)$.
In particular, we may choose $\chi\equiv 1$ on a neighborhood of $x$.
Then 
\[
(P\overline G_+f)(x)
=
(PG_+(\chi f))(x)
=
(\chi f)(x)
=
f(x).
\]
This shows \eqref{eq:GreenOp6}.
Moreover,
\begin{align*}
(\overline G_+Pf)(x)
&=
(G_+(\chi\cdot Pf))(x) \\
&=
(G_+P(\chi f))(x) + (G_+([\chi,P]f))(x)\\
&=
f(x) + (G_+([\chi,P]f))(x).
\end{align*}
In order to prove \eqref{eq:GreenOp5} we have to show $x\notin \supp(G_+([\chi,P]f))$.
The coefficients of the differential operator $[\chi,P]$ vanish where $\chi\equiv1$, hence in particular on $\supp f\cap J^-(x)$.
Now we find
\begin{align*}
\supp(G_+([\chi,P]f))
&\subset
J^+(\supp([\chi,P]f)) \\
&\subset
J^+(\supp f \setminus J^-(x)) \\
&\subset
J^+(\supp f) \setminus \{x\}
\end{align*}
and therefore $x\notin \supp(G_+([\chi,P]f))$.

(f)
As to \eqref{eq:GreenOp7} we see for $f\in\Coopc(M,E_2)$
\begin{align*}
\supp(\overline G_+f)
\subset
\bigcup_{\chi} \supp(G_+(\chi f)) 
\subset
\bigcup_{\chi} J^+(\supp(\chi f)) 
\subset
J^+(\supp f).
\end{align*}
Here the union is taken over all $\chi\in\Cooc(M,\R)$.

(g)
Since the causal future of a past-compact set is again past compact, \eqref{eq:GreenOp7} shows that $\overline G_+$ maps sections with past-compact support to sections with past-compact support.
Now \eqref{eq:GreenOp5} and \eqref{eq:GreenOp6} show that $P$ considered as an operator $\Coopc(M,E_1) \to \Coopc(M,E_2)$ is bijective and that $\overline G_+$ is its inverse.
In particular, $\overline G_+$ is uniquely determined.
\end{proof}

\begin{cor}\label{cor:sol}
There are no nontrivial solutions $f\in\Coo(M,E_1)$ of the differential equation $Pf=0$ with past-compact or future-compact support.
For any $g\in\Coopc(M,E_2)$ or $g\in\Coofc(M,E_2)$ there exists a unique $f\in\Coo(M,E_1)$ solving $Pf=g$ and such that $\supp(f)\subset J^+(\supp(g))$ or $\supp(f)\subset J^-(\supp(g))$, respectively.
\hfill\qed
\end{cor}

Since the causal future of a strictly past-compact set is again strictly past compact we can restrict $\overline G_+$ to smooth sections with strictly past-compact support and we get

\begin{cor}\label{cor:Gtilde}
There are unique linear extensions
\begin{align*}
\tilde G_+:\Coospc(M,E_2) \to \Coospc(M,E_1) \quad\mbox{ and }\quad
\tilde G_-:\Coosfc(M,E_2) \to \Coosfc(M,E_1)
\end{align*}
of $G_+$ and $G_-$ respectively, such that 
\begin{enumerate}[(i)]
\item\label{eq:GreenOp8}
$\tilde G_+Pf=f$ for all $f\in\Coospc(M,E_1)$;
\item\label{eq:GreenOp9}
$P\tilde G_+f=f$ for all $f\in\Coospc(M,E_2)$;
\item\label{eq:GreenOp10}
$\supp(\tilde G_+f)\subset J^+(\supp f)$ for all $f\in\Coospc(M,E_2)$;
\end{enumerate}
and similarly for $\tilde G_-$.
\hfill\qed
\end{cor}

\subsection{Uniqueness and continuity of Green's operators}
The extension of Green's operators to sections with past-compact support will now be used to show continuity and uniqueness of the Green's operators.

\begin{cor}
The Green's operators $G_\pm:\Cooc(M,E_2)\to\Coo(M,E_1)$ as well as the extensions
\begin{alignat*}{3}
\tilde G_+&:\Coospc(M,E_2) \to \Coospc(M,E_1),& \quad
\tilde G_-&:\Coosfc(M,E_2) \to \Coosfc(M,E_1), \\
\overline G_+&:\Coopc(M,E_2) \to \Coopc(M,E_1),& \quad
\overline G_-&:\Coofc(M,E_2) \to \Coofc(M,E_1)
\end{alignat*}
are continuous.
\end{cor}

\begin{proof}
The operator $\overline G_+:\Coopc(M,E_2) \to \Coopc(M,E_1)$ is the inverse of $P$ when considered as an operator $\Coopc(M,E_1) \to \Coopc(M,E_2)$.
If $A\in pc$, then also $J^+(A)\in pc$.
Now $\overline G_+$ maps sections with support in $J^+(A)$ to sections with support in $J^+(J^+(A))=J^+(A)$.
Hence $P$ yields a bijective linear operator $\Coo_{J^+(A)}(M,E_1) \to \Coo_{J^+(A)}(M,E_2)$ with inverse given by the restriction of $\overline G_+$ to $\Coo_{J^+(A)}(M,E_2)$.
By the open mapping theorem for Fr\'echet spaces \cite[Cor.~1, p.~172]{T06}, $\overline G_+$ is continuous as a map $\Coo_{J^+(A)}(M,E_2) \to \Coo_{J^+(A)}(M,E_1)$.
Since we have the continuous embeddings $\CooA(M,E_2) \subset \Coo_{J^+(A)}(M,E_2)$ and $\Coo_{J^+(A)}(M,E_1)\subset\Coopc(M,E_1)$, the operator $\overline G_+$ is also continuous as a map $\CooA(M,E_2) \to \Coopc(M,E_1)$.
Since this holds for any $A\in pc$, we conclude that $\overline G_+:\Coopc(M,E_2) \to \Coopc(M,E_1)$ is continuous.

A similar argument shows that $\tilde G_+:\Coospc(M,E_2) \to \Coospc(M,E_1)$ is continuous.
Using the continuous embeddings $\Cooc(M,E_2) \subset \Coospc(M,E_2)$ and $\Coospc(M,E_1) \subset \Coo(M,E_1)$ we see that the Green's operator $G_+$ is continuous.
The same reasoning proves the claim for $G_-$, $\tilde G_-$, and $\overline G_-$.
\end{proof}

\begin{cor}\label{cor:GreenUnique}
The Green's operators of a Green-hyperbolic operator are unique.
\end{cor}

\begin{proof}
The advanced Green's operator $G_+$ is a restriction of the operator $\overline G_+$ which is uniquely determined by $P$ (as the inverse of $P:\Coopc(M,E)\to\Coopc(M,E)$) and similarly for $G_-$.
\end{proof}

\subsection{Composition of Green-hyperbolic operators}
We now show that the composition as well as ``square roots'' of Green-hyperbolic operators and again Green hyperbolic.

\begin{cor}\label{cor:Komposition}
Let $P_1:\Coo(M,E_1)\to\Coo(M,E_2)$ and $P_2:\Coo(M,E_2)\to\Coo(M,E_3)$ be Green hyperbolic.
Then $P_2\circ P_1:\Coo(M,E_1)\to\Coo(M,E_3)$ is Green hyperbolic.
\end{cor}

\begin{proof}
Denote the Green's operators of $P_i$ by $G_\pm^i$.
We obtain an advanced Green's operator of $P_2\circ P_1$ by composing the following maps:
\[
\Cooc(M,E_3) \hookrightarrow \Coopc(M,E_3) \xrightarrow{\overline G_+^2} \Coopc(M,E_2) \xrightarrow{\overline G_+^1} \Coopc(M,E_2) \hookrightarrow \Coo(M,E_1)
\]
and similarly for the retarded Green's operator.
\end{proof}

\begin{ex}
Let $P=\Box^2=\frac{\partial^4}{\partial t^4}-2\frac{\partial^4}{\partial x^2t^2}+\frac{\partial^4}{\partial x^4}$ be the square of the d'Alembert operator on $2$-dimensional Minkowski space $M=\{(t,x)\in\R^2\}$.
In \eref{ex:Mink2} be have seen that the integral kernel of the Green's operator $G_+^\Box$ is given by $(G_+f)(t,x)=-\frac12 \int_{J^-(t,x)}f(\tau,\xi)\,d\xi\,d\tau$.
Hence $P$ has the Green's operator
\begin{align*}
(G_+f)(t,x) 
&=
((G_+^\Box)^2 f)(t,x) 
=
\frac14 \int_{J^-(t,x)} \int_{J^-(\tau,\xi)} f(s,y) \,d\tau\,d\xi\,ds\,dy \\
&=
\frac14 \int_M \mathrm{Area}(J^-(t,x)\cap J^+(s,y)) f(s,y) \,ds\,dy .
\end{align*}
The integral kernel $\frac14 \mathrm{Area}(J^-(t,x)\cap J^+(s,y))$ of $G_+$ is a continuous function in this case.
\end{ex}

There is a very useful partial inverse to \cref{cor:Komposition}.

\begin{cor}\label{cor:square}
Let $P:\Coo(M,E)\to\Coo(M,E)$ be a differential operator such that $P^2$ is Green hyperbolic.
Then $P$ itself is Green hyperbolic.
\end{cor}

\begin{proof}
\tref{thm:Gquer} applied to $P^2$ tells us that $P^2$ maps $\Coopc(M,E)$ bijectively onto itself.
Hence $P$ itself also maps $\Coopc(M,E)$ bijectively onto itself.
Let $G_+$ denote the composition $\Cooc(M,E) \hookrightarrow \Coopc(M,E) \xrightarrow{P^{-1}} \Coopc(M,E) \hookrightarrow \Coo(M,E)$.
Then $G_+$ obviously satisfies \eqref{eq:GreenOp1} and \eqref{eq:GreenOp2} in \dref{def:GreenOp}.

As to \eqref{eq:GreenOp3}, let $f\in\Cooc(M,E)$.
Put $A:=J^+(\supp f)\in pc$.
Again by \tref{thm:Gquer}, $P^2$ maps $\CooA(M,E)$ bijectively onto itself.
Hence so does $P$ which implies that $G_+$ maps $\CooA(M,E)$ bijectively onto itself.
In particular, $\supp(G_+f)\subset A = J^+(\supp f)$.

The arguments for $G_-$ and for $\Pt$ are analogous.
\end{proof}

\begin{ex}\label{ex:Dirac1}
A differential operator $P$ of first order is said to be of \emph{Dirac type} if $P^2$ is a wave operator.
Since wave operators are Green hyperbolic, \cref{cor:square} tells us that Dirac-type operators are Green hyperbolic too.
Examples are the classical Dirac operator acting on sections of the spinor bundle $E=SM$ (see \cite{B81} for details) or, more generally, on sections of a twisted spinor bundle $E=SM\otimes F$ where $F$ is any ``coefficient bundle'' equipped with a connection.

Particular examples are the \emph{Euler operator} $P=i(d-\delta)$ on $E=\bigoplus_k \Lambda^kT^*M$ and, in dimension $\dim(M)=4$, the \emph{Buchdahl operators} on $SM\otimes S_+^{\odot k}M$.
See \cite[Sec.~2.5]{BG12b} for details.
\end{ex}

If the vector bundles $E_1,E_2\to M$ carry possibly indefinite but nondegenerate fiber metrics $\<\cdot,\cdot\>$, then the \emph{formally adjoint} operator $P^*$ is characterized by 
\begin{equation}
\int_M \<g,Pf\> \dV = \int_M \<P^*g,f\> \dV
\label{eq:FormalAdjoint}
\end{equation}
for all $f\in\Coo(M,E_1)$ and $g\in\Coo(M,E_2)$ with $\supp f\cap \supp g$ compact.
This definition is similar to that of the formally dual operator in \eqref{eq:FormalDual}.
In \eqref{eq:FormalAdjoint} the brackets $\<\cdot,\cdot\>$ denote fiber metrics while in \eqref{eq:FormalDual} they denote the canonical pairing.

\subsection{Direct sum of Green-hyperbolic operators}
The direct sum of two Green-hyperbolic operators is again Green hyperbolic.

\begin{lem}\label{lem:sum}
Let $P:\Coo(M,E_1) \to \Coo(M,E_2)$ and $Q:\Coo(M,E_1') \to \Coo(M,E_2')$ be Green hyperbolic.
Then the operator
\[
\begin{pmatrix}
P & 0 \\
0 & Q
\end{pmatrix}
: \Coo(M,E_1\oplus E_1') \to  \Coo(M,E_2\oplus E_2')
\]
is also Green hyperbolic.
\end{lem}

\begin{proof}
If $G_\pm$ and $G_\pm'$ are the Green's operators for $P$ and $Q$ respectively, then 
$
\begin{pmatrix}
G_\pm & 0 \\
0 & G_\pm'
\end{pmatrix}
$
yields Green's operators for 
$\begin{pmatrix}
P & 0 \\
0 & Q
\end{pmatrix}
$.
\end{proof}

\begin{rem}
The simple construction in \lref{lem:sum} shows that Green hyperbolicity cannot be read off the principal symbol of the operator.
For instance, $P$ could be a wave operator and $Q$ a Dirac-type operator.
Then the total Green-hyperbolic operator in \lref{lem:sum} is of second order and the principal symbol does not see $Q$ and therefore cannot recognize $Q$ as a Green hyperbolic operator.

For similar reasons, it is not clear how to characterize Green hyperbolicity in terms of well-posedness of a Cauchy problem in general.
\end{rem}

Now we get the following variation of \cref{cor:square} for operators acting on sections of two different bundles:

\begin{cor}\label{cor:P*P}
Let $P:\Coo(M,E_1)\to \Coo(M,E_2)$ be a differential operator and let $E_1$ and $E_2$ carry nondegenerate fiber metrics.
Let $P^*:\Coo(M,E_2)\to \Coo(M,E_1)$ be the formally adjoint operator.

If $P^*P$ and $PP^*$ are Green hyperbolic, then $P$ and $P^*$ are Green hyperbolic too.
\end{cor}

\begin{proof}
Consider the operator $\P:\Coo(M,E_1\oplus E_2) \to \Coo(M,E_1\oplus E_2)$ defined by 
\[
\P = 
\begin{pmatrix}
0 & P^* \\
P & 0
\end{pmatrix} .
\]
Since $P^*P$ are $PP^*$ are Green hyperbolic so is
\[
\P^2 = 
\begin{pmatrix}
P^*P & 0 \\
0 & PP^*
\end{pmatrix} .
\]
By \cref{cor:square}, $\P$ is Green hyperbolic.
Let 
\[
\G_\pm =
\begin{pmatrix}
G_\pm^{11} & G_\pm^{21} \\
G_\pm^{12} & G_\pm^{22}
\end{pmatrix}
\]
be the Green's operators of $\P$.
Then one easily sees that $G_\pm^{21}$ are Green's operators for $P$ and $G_\pm^{12}$ for $P^*$. 
\end{proof}

\begin{ex}\label{ex:Dirac2}
If $M$ is even dimensional, then the spinor bundle splits into ``chirality subbundles'' $SM=S_+M\oplus S_-M$.
The twisted Dirac operators in Example~\ref{ex:Dirac1} interchange these bundles and we get operators $P:\Coo(M,S_+M\otimes F) \to \Coo(M,S_-M\otimes F)$.
By \cref{cor:P*P}, they are Green hyperbolic too.
\end{ex}

\subsection{Green's operators of the dual operator}
Next we show that the Green's operators of the dual operator are the duals of the Green's operators.
The roles of ``advanced'' and ``retarded'' get interchanged.

\begin{lem}\label{lem:Greendual}
Let $P:\Coo(M,E_1)\to\Coo(M,E_2)$ be Green hyperbolic.
Denote the Green's operators of $P$ by $G_\pm$ and the ones of\, $\Pt$ by $G^*_\pm$.
Then
\[
\int_M\<\tilde G_-^*\phi,f\>\dV=\int_M\<\phi,\overline G_+f\>\dV 
\]
holds for all $\phi\in \Coosfc(M,E_1^*)$ and $f\in \Coopc(M,E_2)$. 
Similarly,
\[
\int_M\<\tilde G_+^*\phi,f\>\dV=\int_M\<\phi,\overline G_-f\>\dV 
\]
holds for all $\phi\in \Coospc(M,E_1^*)$ and $f\in \Coofc(M,E_2)$. 
\end{lem}

\begin{proof}
By \eqref{eq:GreenOp6} in \tref{thm:Gquer} we have
\begin{align*}
\int_M\< \tilde G_-^*\phi,f\>\dV
&=\int_M\<\tilde G_-^*\phi,P(\overline G_+f)\>\dV\\
&=\int_M\<\Pt(\tilde G_-^*\phi),\overline G_+ f\>\dV\\
&=\int_M\<\phi,\overline G_+ f\>\dV.
\end{align*}
The integration by parts is justified because the intersection $\supp(\tilde G_-^*\phi)\cap\supp(\overline G_+f)$ of a strictly future-compact set and a past-compact set is compact.
The second assertion is analogous.
\end{proof}

\subsection{The causal propagator}
The following theorem contains important information about the solution theory of Green-hyperbolic operators.
It was proved in \cite[Thm.~3.5]{BG12b}, compare also \tref{thm:exactDistri}.

\begin{thm}\label{thm:exact}
Let $G$ be the causal propagator of the Green-hyperbolic operator $P:\Coo(M,E_1)\to\Coo(M,E_2)$.
Then
\begin{equation}\label{eq:exactseq}
 \{0\}\to\Cooc(M,E_1)\xrightarrow{P}\Cooc(M,E_2)\xrightarrow{G}\Coosc(M,E_1)\xrightarrow{P}\Coosc(M,E_2)
\end{equation}
is an exact sequence.
\hfill\qed
\end{thm}

%%%%%%%%%%%%%%%%%%%%%%%%%%%%%%%%%%%%%%%%%%%%%%%%%%%%%%%%%%%%%%%%%%%%%%%%%%%%%%%%%%%%%%%%%%%%%%%%%

\section{Green-hyperbolic Operators Acting on Distributional Sections}
\label{sec:Distri}

\subsection{Green's operators acting on distributional sections}
We extend any differential operator $P:\Coo(M,E_1)\to\Coo(M,E_2)$ as usual to distributional sections by taking the dual map of $\Pt:\Cooc(M,E_2^*)\to\Cooc(M,E_1^*)$
thus giving rise to a continuous linear map $P:\D(M,E_1)\to\D(M,E_2)$.

\begin{lem}\label{lem:Ghut}
The Green's operators $\overline G_+:\Coopc(M,E_2) \to \Coopc(M,E_1)$ and $\overline G_-:\Coofc(M,E_2) \to \Coofc(M,E_1)$ extend uniquely to continuous operators
\[
\widehat G_+:\Dpc(M,E_2) \to \Dpc(M,E_1)
\mbox{ and }
\widehat G_-:\Dfc(M,E_2) \to \Dfc(M,E_1),
\]
respectively.
Moreover
\begin{enumerate}[(i)]
\item\label{eq:GreenOp11}
$\widehat G_+Pf=f$ holds for all $f\in\Dpc(M,E_1)$;
\item\label{eq:GreenOp12}
$P\widehat G_+f=f$ holds for all $f\in\Dpc(M,E_2)$;
\item\label{eq:GreenOp13}
$\supp(\widehat G_+f)\subset J^+(\supp f)$ holds for all $f\in\Dpc(M,E_2)$;
\end{enumerate}
and similarly for $\widehat G_-$.
\end{lem}

\begin{proof}
Recall from \lref{lem:DistriSupp} and Table~1 that $\Dpc(M,E_i)$ can be identified with the dual space of $\Coosfc(M,E_i^*)$.
Let $\widehat G_+$ be the dual map of $\tilde G_-^*:\Coosfc(M,E_1^*)\to\Coosfc(M,E_2^*)$ where $G_-^*$ is the retarded Green's operator of $\Pt$.
By \lref{lem:Greendual}, $\widehat G_+$ is an extension of $\overline G_+$.
The extension is unique because $\Cooc(M,E_2)$ is dense in $\Dpc(M,E_2)$ by \lref{lem:Ccdicht2}.

Dualizing \eqref{eq:GreenOp8} and \eqref{eq:GreenOp9} for $\Pt$ and $G_-^*$ in \cref{cor:Gtilde} we get \eqref{eq:GreenOp11} and \eqref{eq:GreenOp12} as asserted.
As to \eqref{eq:GreenOp13} let $f\in\Dpc(M,E_2)$ and let $\phi\in\Cooc(M,E_2^*)$ be a test section such that $J^+(\supp f)\cap\supp(\phi)=\emptyset$.
Then $\supp f\cap J^-(\supp(\phi))=\emptyset$ and therefore
\begin{align*}
(\widehat G_+f)[\phi] = f[\overline G_-^*\phi] =0.
\end{align*}
Thus $\supp(\widehat G_+f)\subset J^+(\supp f)$.
\end{proof}

Summarizing \tref{thm:Gquer}, \cref{cor:Gtilde} and \lref{lem:Ghut} we get the following diagram of continuous extensions of the Green's operator $G_+$ of $P$:
\begin{diagram}\label{diag:GreenExtensions}
\begin{equation*}
\xymatrix{
\Cooc(M,E_2) \ar@{^{(}->}[r] \ar@{=}[d] \ar@/^30pt/[rrr]_{G_+} & \Coospc(M,E_2) \ar[r]^{\tilde G_+}  \ar@{_{(}->}[d] & \Coospc(M,E_1) \ar@{_{(}->}[d]  \ar@{^{(}->}[r]& \Coo(M,E_1) \ar@{=}[d]\\
\Cooc(M,E_2)  \ar@{^{(}->}[r] \ar@{_{(}->}[d]& \Coopc(M,E_2) \ar[r]^{\overline G_+} \ar@{_{(}->}[d]& \Coopc(M,E_1) \ar@{^{(}->}[r] \ar@{_{(}->}[d] & \Coo(M,E_1) \ar@{_{(}->}[d] \\
\Dc(M,E_2)\ar@{^{(}->}[r]  & \Dpc(M,E_2) \ar[r]^{\widehat G_+} & \Dpc(M,E_1) \ar@{^{(}->}[r] &  \D(M,E_1)
}
\end{equation*}
\end{diagram}
\diagramnumber{Extensions of the advanced Green's operator}

By \eqref{eq:GreenOp13} in \lref{lem:Ghut}, $\widehat G_+$ also restricts to an operator $\Dspc(M,E_2)\to\Dspc(M,E_1)$.

\cref{cor:sol} holds also for distributional sections:

\begin{cor}\label{cor:solDistri}
There are no nontrivial distributional solutions $f\in\D(M,E_1)$ of the differential equation $Pf=0$ with past-compact or future-compact support.
For any $g\in\Dpc(M,E_2)$ or $g\in\Dfc(M,E_2)$ there exists a unique $f\in\D(M,E_1)$ solving $Pf=g$ and such that $\supp(f)\subset J^+(\supp(g))$ or $\supp(f)\subset J^-(\supp(g))$, respectively.
\hfill\qed
\end{cor}

\subsection{The causal propagator}
Using the restriction of $\widehat G_+$ to an operator $\Dc(M,E_2)\to\Dspc(M,E_1)\hookrightarrow\Dsc(M,E_1)$ and  $\widehat G_-:\Dc(M,E_2)\to\Dsc(M,E_1)$ we obtain an extension of the causal propagator $G:\Cooc(M,E_2) \to \Coosc(M,E_1)$ to distributions:
\[
\widehat G := \widehat G_+-\widehat G_-: \Dc(M,E_2) \to \Dsc(M,E_1). 
\]

Now we get the analog to \tref{thm:exact}.

\begin{thm}\label{thm:exactDistri}
The sequence
\begin{equation}\label{eq:exactseq2}
 \{0\}\to\Dc(M,E_1)\xrightarrow{P}\Dc(M,E_2)\xrightarrow{\widehat G}\Dsc(M,E_1)\xrightarrow{P}\Dsc(M,E_2)
\end{equation}
is exact.
\hfill\qed
\end{thm}

\begin{proof}
It is clear from \eqref{eq:GreenOp11} and \eqref{eq:GreenOp12} in \lref{lem:Ghut} that $P\widehat G=\widehat GP=0$ on $\Dc(M,S)$, hence \eqref{eq:exactseq2} is a complex.

In \cref{cor:solDistri} we have seen that $P$ is injective on $\Dpc(M,E_1)$.
Hence $P$ is injective on $\Dc(M,E_1)$ and the complex is exact at $\Dc(M,E_1)$.

Let $f\in\Dc(M,E_2)$ with $\widehat Gf=0$, i.e., $\widehat G_+f=\widehat G_-f$. 
We put $g:=\widehat G_+f=\widehat G_-f \in \D(M,S)$ and we see that $\supp(g)=\supp(\widehat G_+f)\cap\supp(\widehat G_-f)\subset J^+(\supp(f))\cap 
J^-(\supp(f))$.
Since $J^+(\supp(f))\cap J^-(\supp(f))$ is compact, $g\in \Dc(M,E_1)$.
From $Pg = P\widehat G_+f = f$ we see that $f\in P(\Dc(M,E_2))$.
This shows exactness at $\Dc(M,E_2)$.

It remains to show that any $f\in \Dsc(M,E_1)$ with $Pf=0$ is of the form $f=\widehat Gg$ for some $g\in \Dc(M,E_2)$.
Using a cutoff function decompose $f$ as $f=f_+-f_-$ where $\supp(f_\pm) \subset J^\pm(K)$ where $K$ is a suitable compact subset of $M$.
Then $g:=Pf_+=Pf_-$ satisfies $\supp(g)\subset J^+(K)\cap J^-(K)$.
Thus $g\in \Dc(M,E_2)$.
We check that $\widehat G_+g=f_+$.
Namely, for all $\phi\in\Cooc(M,E_1^*)$ we have by the definition of $\widehat G_+$,
$$
\widehat G_+Pf_+[\phi] =
Pf_+[G_-^{*}\phi] =
f_+[\Pt G_-^{*}\phi] =
f_+[\phi] .
$$
The second equality is justified because $\supp(f_+) \cap \supp(G^*_-\phi) \subset J^+(K)\cap J^-(\supp(\phi))$ is compact.
Similarly, one shows $\widehat G_-g=f_-$.
Now $\widehat Gg = \widehat G_+g - \widehat G_-g = f_+ - f_- = f$ which concludes the proof. 
\end{proof}

%%%%%%%%%%%%%%%%%%%%%%%%%%%%%%%%%%%%%%%%%%%%%%%%%%%%%%%%%%%%%%%%%%%%%%%%%%%%%%%%%%%%%%%%%%%%%%%%%

\section{Symmetric Hyperbolic Systems}
\label{sec:Symm}

\subsection{Definition and example}
Now we consider an important class of operators of first order on Lorentzian manifolds, the symmetric hyperbolic systems.
We will show that the Cauchy problem for such operators is well posed on globally hyperbolic manifolds.
We will deduce that they are Green hyperbolic so that the results of the previous sections apply.
For an approach based on the framework of hyperfunctions see \cite{Sch13}.

For a linear first-order operator $P:\Coo(M,E)\to\Coo(M,F)$ the principal symbol $\sigma_P: T^*M\otimes E\to F$ can be characterized by $P(fu)=fPu + \sigma_P(df)u$ where $u\in\Coo(M,E)$ and $f\in\Coo(M,\R)$.

\begin{dfn}\label{def:symhyp}
Let $M$ be a time oriented Lorentzian manifold.
Let $E\to M$ be a real or complex vector bundle with a (possibly indefinite) nondegenerate sesquilinear fiber metric $\<\cdot,\cdot\>$.
A linear differential operator $P:\Coo(M,E)\to\Coo(M,E)$ of first order is called a \emph{symmetric hyperbolic system over $M$} if the following holds for every $x\in M$:
\begin{enumerate}[(i)]
\item\label{eq:symhyp1}
The principal symbol $\sigma_P(\xi):E_x\to E_x$ is symmetric or Hermitian with respect to $\<\cdot,\cdot\>$ for every $\xi\in T_x^*M$;
\item\label{eq:symhyp2}
For every future-directed timelike covector $\tau\in T_x^*M$, the bilinear form $\<\sigma_P(\tau)\cdot,\cdot\>$ on $E_x$ is positive definite.
\end{enumerate}
\end{dfn}

The first condition relates the principal symbol of $P$ to the fiber metric on $E$, the second relates it to the Lorentzian metric on $M$.
The Lorentzian metric enters only via its conformal class because this suffices to specify the causal types of (co)vectors.

\begin{ex}
Let $M=\R^{n+1}$ and denote generic elements of $M$ by $x=(x^0,x^1,\ldots,x^n)$.
We provide $M$ with the Minkowski metric $g=-(dx^0)^2+(dx^1)^2+\ldots+(dx^n)^2$.
The coordinate function $t=x^0/c:M\to\R$ is a Cauchy temporal function;
here $c$ is a positive constant to be thought of as the speed of light.

Let $E$ be the trivial real or complex vector bundle of rank $N$ over $M$ and let $\<\cdot,\cdot\>$ denote the standard Euclidean scalar product on the fibers of $E$, canonically identified with $\K^N$ where $\K=\R$ or $\K=\C$.
Any linear differential operator $P:\Coo(M,E)\to\Coo(M,E)$ of first order is of the form 
\[
P =  A_0(x) \frac{\partial}{\partial t} + \sum_{j=1}^n A_j(x) \frac{\partial}{\partial x^j} + B(x)
\]
where the coefficients $A_j$ and $B$ are $N\times N$-matrices depending smoothly on $x$.
Condition~\eqref{eq:symhyp1} in \dref{def:symhyp} means that all matrices $A_j(x)$ are symmetric if $\K=\R$ and Hermitian if $\K=\C$.
Condition~\eqref{eq:symhyp2} with $\tau=dt$ means that $A_0(x)$ is positive definite.
Thus $P$ is a symmetric hyperbolic system in the usual PDE sense, see e.g.\ \cite[Def.~2.11]{A09}.
But \eqref{eq:symhyp2} says more than that;
it means that $A_0(x)$ dominates $A_1(x),\ldots,A_n(x)$ in the following sense:
The covector $\tau=dt+\sum_{j=1}^n \alpha_j dx^j$ is timelike if and only if $\sum_{j=1}^n\alpha_j^2 < c^{-2}$.
Thus the matrix
\[
\sigma_P(\tau) = A_0(x) + \sum_{j=1}^n \alpha_j A_j(x)
\]
must be positive definite whenever $\sum_{j=1}^n\alpha_j^2 < c^{-2}$.
\cref{cor:finitespeed} below will tell us that waves $u$ solving the equation $Pu=0$ will propagate at most with speed $c$.
\end{ex}

Many examples important in mathematical physics can be found in \cite[App.~A]{G96}.

Given a first order operator $P$ which is not symmetric hyperbolic, one can still try to find a fiberwise invertible endomorphism field $A\in\Coo(M,\Hom(E,E))$ such that $Q=A\circ P$ is symmetric hyperbolic.
Then the analytic results below apply to $Q$ and hence yield analogous results for $P$ as well.
Finding such an endomorphism field is an algebraic problem which is treated e.g.\ in \cite{S07}.

\subsection{The energy estimate}
The following energy estimate will be crucial for controlling the support of solutions to symmetric hyperbolic systems.
It will establish finiteness of the speed of propagation and uniqueness of solutions to the Cauchy problem.

Let $M$ be globally hyperbolic and let $t:M\to\R$ be a Cauchy temporal function.
We write $\Sigma_s:=t^{-1}(s)$ and $\Sigma_s^x := J^-(x)\cap \Sigma_s$ for $x\in M$.
The scalar product $\<\cdot,\cdot\>_0 := \sqrt{\beta}\<\sigma_P(dt)\cdot,\cdot\>$ is positive definite.
Here the smooth positive function $\beta:M\to\R$ is chosen for normalization, more precisely, the Lorentzian metric on $M$ is given by $g=-\beta dt^2 + g_t$ where each $g_s$ is the induced Riemannian metric on $\Sigma_s$.
Let $\dAs$ be the volume density of $\Sigma_s$.
We denote the norm corresponding to $\<\cdot,\cdot\>_0$ by $|\cdot|_0$.

\begin{thm}[Energy estimate]\label{thm:EnEst}
Let $M$ be globally hyperbolic, let $P$ be a symmetric hyperbolic system over $M$ and let $t:M\to\R$ be a Cauchy temporal function.
For each $x\in M$ and each $t_0\in t(M)$ there exists a constant $C>0$ such that
\[
\int_{\Sigma_{t_1}^x} |u|_0^2\,\dA_{t_1}
\le
\bigg[C\int_{t_0}^{t_1} \int_{\Sigma_s^x} |Pu|_0^2\,\dAs\,ds + \int_{\Sigma_{t_0}^x} |u|_0^2\,\dA_{t_0}\bigg]
 e^{C(t_1-t_0)}
\]
holds for each $u\in\Coo(M,E)$ and for all $t_1\ge t_0$.
\end{thm}

\begin{proof}
Denote the dimension of $M$ by $n+1$.
Without loss of generality, we assume that $M$ is oriented;
if $M$ is nonorientable replace the $(n+1)$- and $n$-forms occurring below by densities or, alternatively, work on the orientation covering of $M$.

Let $\vol$ be the volume form of $M$.
We define the $n$-form $\omega$ on $M$ by 
\[
\omega := \sum_{j=0}^n \Re(\<\sigma_P(b_j^*)u,u\>)\, b_j \lrcorner \vol .
\]
Here $b_0,\ldots,b_n$ denotes a local tangent frame, $b_0^*,\ldots,b_n^*$ the dual basis, and $\lrcorner$ denotes the insertion of a tangent vector into the first slot of a form.
It is easily checked that $\omega$ does not depend on the choice of $b_0,\ldots,b_n$.
For the sake of brevity, we write
\begin{equation}
f:=Pu .
\label{eq:Puf}
\end{equation} 

We choose a metric connection $\nabla$ on $E$.
The symbol $\nabla$ will also be used for the Levi-Civita connection on $TM$.
Since the first-order operator $\sum_{j=0}^n \sigma_P(b_j^*)\nabla_{b_j}$ has the same principal symbol as $P$, it differs from $P$ only by a zero-order term.
Thus there exists $B\in\Coo(M,\Hom(E,E))$ such that 
\begin{equation}
P = \sum_{j=0}^n \sigma_P(b_j^*)\nabla_{b_j} - B .
\label{eq:PRahmen}
\end{equation}
To simplify the computation of the exterior differential of $\omega$, we assume that the local tangent frame is synchronous at the point under consideration, i.e., $\nabla b_j=0$ at the (fixed but arbitrary) point.
In particular, the Lie brackets $[b_j,b_k]$ vanish at that point.
Then we get at that point
\begin{align*}
d\omega(b_0,\ldots,b_n)
&=
\sum_{k=0}^n(-1)^k \partial_{b_k}(\omega(b_0,\ldots,\widehat b_k,\ldots,b_n)) \\
&=
\sum_{k=0}^n(-1)^k \partial_{b_k}\bigg( \sum_{j=0}^n \Re(\<\sigma_P(b_j^*)u,u\>)\, \vol(b_j,b_0,\ldots,\widehat b_k,\ldots,b_n)  \bigg) \\
&=
\Re\sum_{j=0}^n \partial_{b_j}( \<\sigma_P(b_j^*)u,u\>)\, \vol(b_0,\ldots,b_n)
\end{align*}
and thus 
\[
d\omega
=
\Re\sum_{j=0}^n \partial_{b_j}( \<\sigma_P(b_j^*)u,u\>)\, \vol.
\]
We put $\widetilde B:= \sum_{j=0}^n \nabla_{b_j}\sigma_P(b_j^*)\in\Coo(M,\Hom(E,E))$.
Using the symmetry of the principal symbol, \eqref{eq:Puf}, and \eqref{eq:PRahmen} we get
\begin{align*}
\sum_{j=0}^n \partial_{b_j}( \<\sigma_P(b_j^*)u,u\>)
&=
\<\widetilde B u,u\> + \sum_{j=0}^n[\<\sigma_P(b_j^*)\nabla_{b_j}u,u\> + \<\sigma_P(b_j^*)u,\nabla_{b_j}u\> ]\\
&=
\<\widetilde B u,u\> + \<(P+B)u,u\> + \<u,(P+B)u\> \\
&=
\<(\widetilde B +B)u,u\> + \<u,Bu\> + \<f,u\> + \<u,f\>
\end{align*}
and hence
\[
d\omega
=
\Re(\<(\widetilde B +2B)u,u\> + 2\<f,u\>)\,\vol .
\]
Thus we have for any compact $K\subset M$
\begin{align*}
\int_K d\omega
&=
\int_K \Re(\<(\widetilde B +2B)u,u\> + 2\<f,u\>)\,\vol \\
&\leq
\int_K (C_1 |u|_0^2 + C_2|f|_0|u|_0)\,\vol \\
&\leq
C_3\int_K ( |u|_0^2 + |f|_0^2)\,\vol
\end{align*}
with constants $C_1$, $C_2$, $C_3$ depending on $P$ and $K$ but not on $u$ and $f$.
We apply this to $K=J^-(x)\cap t^{-1}([t_0,t_1])$ where $[t_0,t_1]$ is a compact subinterval of the image of $t$ (Fig.~2).
\begin{center}
\begin{pspicture}(-4,0.8)(4,5)
%\rput(0,3){\psscalebox{0.15}{\epsfbox{Kegel-gross.eps}}}
\rput(0,3){\psscalebox{0.4}{\epsfbox{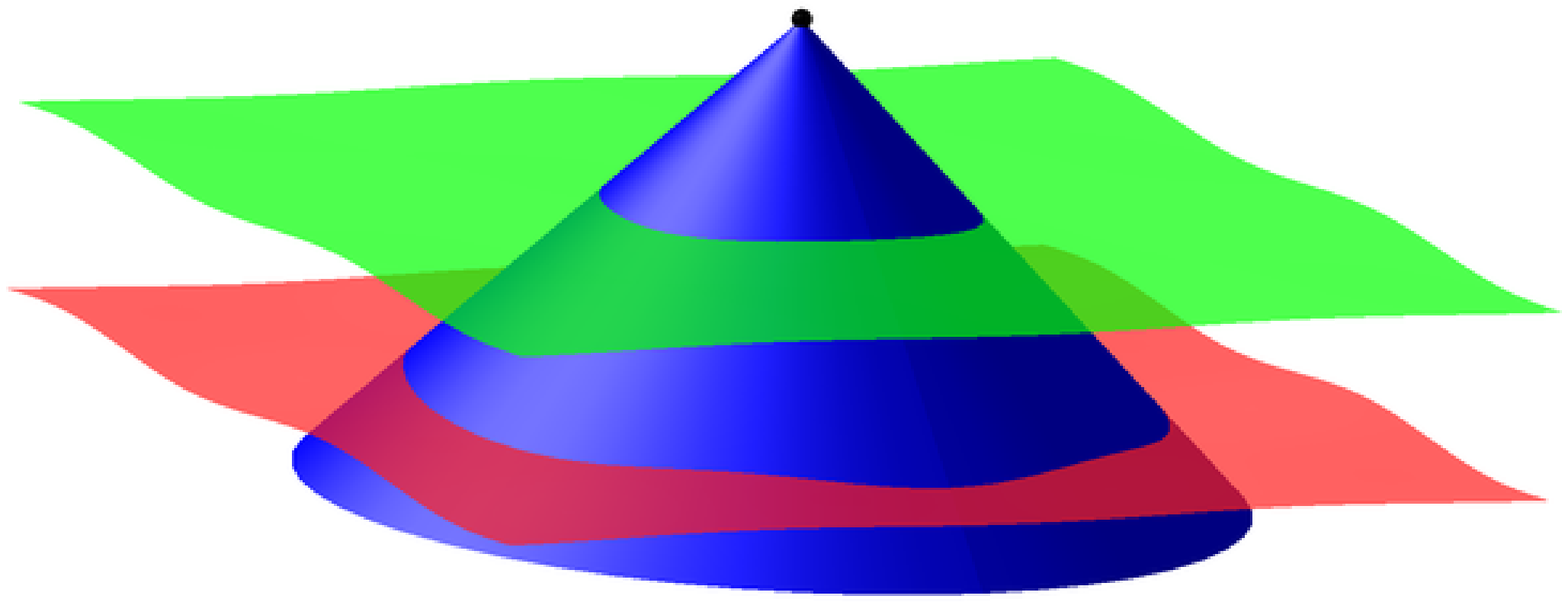}}}
\rput(0.1,4.8){$x$}
\rput(-4.5,4){$\Sigma_{t_1}$}
\rput(-4.5,3){$\Sigma_{t_0}$}
\rput(1,1){$J^-(x)$}
\rput(0,2.4){\textcolor{white}{$K$}}
\end{pspicture}

{\sc Fig.}~2:
\emph{Integration domain in the energy estimate}
\end{center}

By the Fubini theorem,
\begin{equation}
\int_K d\omega
\le
C_4\int_{t_0}^{t_1} \int_{\Sigma_s^x} ( |u|_0^2 + |f|_0^2)\,\dAs\,ds \, .
\label{eq:domegabound}
\end{equation}
The boundary $\partial J^-(x)$ is a Lipschitz hypersurface (see \cite[p.~187]{HE73} or \cite[pp.~413--415]{ON83}).
The Stokes' theorem for manifolds with Lipschitz boundary \cite[p.~209]{EG92} yields
\begin{equation}
\int_K d\omega
=
\int_{\partial K}\omega
=
\int_{\Sigma_{t_1}^x}\omega - \int_{\Sigma_{t_0}^x}\omega + \int_Y\omega 
\label{eq:Stokes}
\end{equation}
where $Y=(\partial J^-(x))\cap t^{-1}([t_0,t_1])$.
Choosing $b_0=\sqrt{\beta}dt$ and $b_1,\ldots,b_n$ tangent to $\Sigma_s$, we see that
\begin{equation}
\int_{\Sigma_s^x}\omega
=
\int_{\Sigma_s^x} \<\sigma_P(\sqrt{\beta}dt)u,u\>\,\dAs
=
\int_{\Sigma_s^x} |u|_0^2\,\dAs .
\label{eq:Deckel}
\end{equation}
The boundary $\partial J^-(x)$ is ruled by the past-directed lightlike geodesics emanating from $x$.
Thus at each differentiable point $y\in \partial J^-(x)$ the tangent space $T_y\partial J^-(x)$ contains a lightlike vector but no timelike vectors.
We choose a positively oriented generalized orthonormal tangent basis $b_0,b_1,\ldots,b_n$ of $T_yM$ in such a way that $b_0$ is future-directed timelike and $b_0+b_1,b_2,\ldots,b_n$ is a oriented basis of $T_y\partial J^-(x)$.
Then
\begin{align*}
\omega(b_0+b_1,b_2,\ldots,b_n)
&=
\sum_{j=0}^n \Re(\<\sigma_P(b_j^*)u,u\>)\vol(b_j,b_0+b_1,b_2,\ldots,b_n) \\
&=
\Re\<\sigma_P(b_0^*)u,u\>-\Re\<\sigma_P(b_1^*)u,u\> \\
&=
\Re\<\sigma_P(b_0^*-b_1^*)u,u\> .
\end{align*}
Since $\<\sigma_P(\tau)\cdot,\cdot\>$ is positive definite for each future-directed timelike covector, it is, by continuity, still positive semidefinite for each future-directed causal covector.
Now $b_0^*-b_1^*$ is future-directed lightlike.
Therefore
\[
\omega(b_0+b_1,b_2,\ldots,b_n) = \<\sigma_P(b_0^*-b_1^*)u,u\> \ge 0.
\]
This implies 
\begin{equation}
\int_Y\omega\ge 0.
\label{eq:Yge0}
\end{equation}
Combining \eqref{eq:domegabound}, \eqref{eq:Stokes}, \eqref{eq:Deckel}, and \eqref{eq:Yge0} we find
\[
\int_{\Sigma_{t_1}^x} |u|_0^2\,\dA_{t_1} - \int_{\Sigma_{t_0}^x} |u|_0^2\,\dA_{t_0}
\le
C_4\int_{t_0}^{t_1} \int_{\Sigma_s^x} ( |u|_0^2 + |f|_0^2)\,\dAs\,ds \, .
\]
In other words, the function $h(s)=\int_{\Sigma_{s}^x} |u|_0^2\,\dAs$ satisfies the integral inequality
\[
h(t_1) \leq \alpha(t_1) + C_4\int_{t_0}^{t_1}h(s)\,ds
\]
for all $t_1\ge t_0$ where $\alpha(t_1) = C_4\int_{t_0}^{t_1} \int_{\Sigma_s^x} |f|_0^2\,\dAs\,ds + h(t_0)$.
Gr\"onwall's lemma gives
\[
h(t_1)\le \alpha(t_1) e^{C_4(t_1-t_0)}
\]
which is the claim.
\end{proof}

\subsection{Finite speed of propagation}
We deduce that a ``wave'' governed by a symmetric hyperbolic system can propagate with the speed of light at most (Fig.~3).

\begin{cor}[Finite propagation speed]\label{cor:finitespeed}
Let $M$ be globally hyperbolic, let $\Sigma\subset M$ be a smooth spacelike Cauchy hypersurface and let $P$ be a symmetric hyperbolic system over $M$.
Let $u\in\Coo(M,E)$ and put $u_0:=u|_\Sigma$ and $f:=Pu$.
Then
\begin{equation}
\supp(u)\cap J^\pm(\Sigma) \subset J^\pm((\supp f\cap J^\pm(\Sigma)) \cup \supp\, u_0)  .
\label{eq:finitefuturepast}
\end{equation}
\begin{center}
\begin{pspicture}(-3.4,-1.1)(4,2.4)
\pspolygon[fillcolor=lightgray,fillstyle=solid,linewidth=0.01](-3,2)(-1,0)(0,0)(.5,.5)(1,0)(3,0)(5,2)
\psline[linecolor=white](-3,2)(5,2)
\psline[linewidth=0.02](-3.5,0)(5.5,0)
\psline[linewidth=0.06](-1,0)(0,0)
\psccurve[fillcolor=gray,fillstyle=solid,linewidth=0.01](1,0)(2,1)(3,0)(2.5,-0.6)(1.5,-0.7)
\psline[linewidth=0.01,linestyle=dotted,dotsep=0.03](1,0)(3,0)
\rput(-3.8,0){$\Sigma$}
\rput(-3.8,1.3){$J^+(\Sigma)$}
\rput(-0.5,-0.4){$\supp(u_0)$}
\rput(2,0){\psframebox*[framearc=.3]{$\supp f$}}
\end{pspicture}

{\sc Fig.}~3:
\emph{Finite propagation speed}
\end{center}
In particular,
\[
\supp(u) \subset J(\supp f \cup \supp(u_0)).
\]
\end{cor}

\begin{proof}
One can choose a Cauchy temporal function in such a way that $\Sigma=\Sigma_0$ where again $\Sigma_s=t^{-1}(s)$, see \cite[Thm.~1.2~(B)]{BS06}.
Let $x\in J^+(\Sigma)$.
Assume $x\in M\setminus J^+((\supp f\cap J^+(\Sigma)) \cup \supp(u_0))$.
This means that there is no future-directed causal curve starting in $\supp f \cup \supp\, u_0$, entirely contained in $J^+(\Sigma)$, which terminates at $x$.
In other words, there is no past-directed causal curve starting at $x$, entirely contained in $J^+(\Sigma)$, which terminates in $\supp f \cup \supp\, u_0$.
Hence $J^-(x)\cap J^+(\Sigma)$ does not intersect $\supp f \cup \supp(u_0)$.
By \tref{thm:EnEst}, $u$ vanishes on $J^-(x)\cap J^+(\Sigma)$, in particular $u(x)=0$.
This proves \eqref{eq:finitefuturepast} for $J^+$.

The case $x\in J^-(\Sigma)$ can be reduced to the previous case by time reversal.
For the support of $u$ we deduce
\begin{align*}
\supp\, u
&\subset
J^+((\supp f\cap J^+(\Sigma)) \cup \supp\, u_0) \cup
J^-((\supp f\cap J^-(\Sigma)) \cup \supp\, u_0) \\
&\subset
J^+(\supp f \cup \supp\, u_0) \cup
J^-(\supp f \cup \supp\, u_0) \\
&=
J(\supp f \cup \supp\, u_0).\qedhere
\end{align*}

\end{proof}

\subsection{Uniqueness of solutions to the Cauchy problem}
As a consequence we obtain uniqueness for the Cauchy problem.

\begin{cor}[Uniqueness for the Cauchy problem]\label{cor:CPuniqueness}
Let $M$ be globally hyperbolic, let $\Sigma\subset M$ be a smooth spacelike Cauchy hypersurface and let $P$ be a symmetric hyperbolic system over $M$.
Given $f\in\Coo(M,E)$ and $u_0\in\Coo(\Sigma,E)$ there is at most one solution $u\in\Coo(M,E)$ to the Cauchy problem
\begin{equation}
\begin{cases}
Pu=f ,\\
u|_\Sigma=u_0 .
\end{cases}
\label{eq:CauchyProblem}
\end{equation}
\end{cor}

\begin{proof}
By linearity, we only need to consider the case $f=0$ and $u_0=0$.
In this case, \cref{cor:finitespeed} shows $\supp\, u\subset J(\emptyset)=\emptyset$, hence $u=0$.
\end{proof}

\subsection{Existence of solutions to the Cauchy problem}
Existence of solutions is obtained by gluing together local solutions.
The latter exist due to standard PDE theory.
A uniqueness and existence proof for local solutions to quasilinear hyperbolic systems was also sketched in \cite[App.~B]{G96}.
It should be noted that global hyperbolicity of the underlying manifold is still crucial.
It is needed in order to have several compactness properties used in the proof of \tref{thm:CPexistence}.

\begin{thm}[Existence for the Cauchy problem]\label{thm:CPexistence}
Let $M$ be globally hyperbolic, let $\Sigma\subset M$ be a smooth spacelike Cauchy hypersurface and let $P$ be a symmetric hyperbolic system over $M$.
For any $f\in\Coo(M,E)$ and any $u_0\in\Coo(\Sigma,E)$ there is a solution $u\in\Coo(M,E)$ to the Cauchy problem \eqref{eq:CauchyProblem}.
\end{thm}

\begin{proof}
(a)
We first assume that $u_0$ and $f$ have compact supports.
We reduce the existence statement to standard PDE theory.
Choose a Cauchy temporal function $t:M\to\R$ with $\Sigma_0=\Sigma$.
Write $t(M)=(t_-,t_+)$ where $-\infty\leq t_- < 0 < t_+ \leq \infty$.
Put 
\[
t_*:=\sup\{\tau\in[0,t_+] \mid \mbox{there exists a $\Coo$-solution $u$ to \eqref{eq:CauchyProblem} on $t^{-1}([0,\tau))$}\} .
\]
We have to show $t_*=t_+$.
Assume $t_*<t_+$.
For each $\tau<t_*$ we have a solution of \eqref{eq:CauchyProblem} on $t^{-1}([0,\tau))$.
By uniqueness, the solutions for different $\tau$'s coincide on their common domain.
Thus we have a solution $u$ on $t^{-1}([0,t_*))$.
Put $K:=\supp(u_0) \cup \supp f$.
We cover the compact set $J(K)\cap\Sigma_{t_*}$ by finitely many causally compatible, globally hyperbolic coordinate charts $U_1,\ldots,U_N$ over which the vector bundle $E$ is trivial.
Choose $\eps>0$ small enough so that the union $U_1\cup\cdots\cup U_N$ still contains $J(K)\cap \Sigma_\tau$ for each $\tau\in[t_*-\eps,t_*+\eps]$.
Choose $\psi_j\in\Cooc(M,\R)$ such that $\supp\,\psi_j\subset U_j$ and
\begin{equation}
\psi_1+\cdots+\psi_N\equiv 1 \quad\mbox{ on }\quad J(K)\cap t^{-1}([t_*-\eps,t_*+\eps]) .
\label{eq:TeilungEins}
\end{equation}

In local coordinates and with respect to a local trivialization of $E$, the operator $P$ is a symmetric hyperbolic system in the classical PDE sense so that we can find local solutions $u_j\in\Coo(U_j,E)$ of the Cauchy problem 
\[
\begin{cases}
Pu_j=\psi_j f,\\
u_j|_{U_j\cap\Sigma_{t_*-\eps}}=\psi_ju|_{U_j\cap\Sigma_{t_*-\eps}},
\end{cases}
\]
see e.g.\ \cite[Thm.~7.11]{A09}. 
By \cref{cor:CPuniqueness}, $\supp\, u_j \subset J(\supp\,\psi_j)$.
Since $\supp\,\psi_j$ is a compact subset of $U_j$, there exists an $\eps_j>0$ such that $J(\supp\,\psi_j)\cap t^{-1}([t_*-\eps_j,t_*+\eps_j]) \subset U_j$.
Thus we can extend $u_j$ by zero to a smooth section, again denoted by $u_j$, defined on $t^{-1}([t_*-\eps_j,t_*+\eps_j])$.
For $\eps_0:=\min\{\eps,\eps_1,\ldots,\eps_N\}$, 
\[
v:=u_1+\cdots+u_N                                                  
\]
is a smooth section defined on $t^{-1}([t_*-\eps_0,t_*+\eps_0])$.
Now
\[
v|_{\Sigma_{t_*-\eps}}
=
(\psi_1+\ldots+\psi_N)\cdot u|_{\Sigma_{t_*-\eps}}
=
u|_{\Sigma_{t_*-\eps}}
\]
because $\supp\,u\subset J(K)$ so that \eqref{eq:TeilungEins} applies.
Moreover, on $t^{-1}([t_*-\eps,t_*+\eps])$, 
\[
Pv
=
Pu_1+\ldots +Pu_N
=
(\psi_1+\ldots+\psi_N)\cdot f
=
f
\]
because $\supp f\subset K\subset J(K)$.
Thus $u$ and $v$ solve the same Cauchy problem on $t^{-1}([t_*-\eps_0,t_*))$ and hence coincide in this region.
Therefore $v$ extends $u$ smoothly to a solution of \eqref{eq:CauchyProblem} on $t^{-1}([0,t_*+\eps])$ which contradicts the maximality of $t_*$.
This shows $t_*=t_+$.
Similarly, one extends the solution to $t^{-1}((t_-,0])$, hence to all of $M$.
%Since the support of $u$ is contained in $J(K)$, it is spacially compact.

(b)
Now we drop the assumption that $u_0$ and $f$ have compact supports.
Let $x\in M$.
Without loss of generality assume $x\in J^+(\Sigma)$.
Since $M$ is globally hyperbolic, $J^-(x)\cap J^+(\Sigma)$ is compact.
We choose a cutoff function $\chi\in\Cooc(M,\R)$ with $\chi\equiv 1$ on an open neighborhood of $J^-(x)\cap J^+(\Sigma)$.
Define $u(x):=u_x(x)$ where $u_x$ is the solution of the Cauchy problem
\[
\begin{cases}
Pu_x=\chi f,\\
u_x|_\Sigma=\chi u_0.
\end{cases}
\]
The solution $u_x$ exists by part (a) of the proof.

We claim that, in a neighborhood of $x$, the solution $u_x$ does not depend on the choice of cutoff function $\chi$.
Namely, let $\tilde\chi$ be another such cutoff function and $\tilde u_x$ the corresponding solution.
Then $v:=u_x - \tilde u_x$ solves the Cauchy problem
\[
\begin{cases}
Pv=(\chi-\tilde\chi) f,\\
v|_\Sigma=(\chi-\tilde\chi) u_0.
\end{cases}
\]
Since $\chi-\tilde\chi$ vanishes on a neighborhood of $J^-(x)\cap J^+(\Sigma)$, $v$ must vanish in a neighborhood of $x$ by \cref{cor:CPuniqueness}.
\hfill\ding{51}

In particular, $u$ is a smooth section which coincides with $u_x$ in a neighborhood of $x$.
Thus we have $(Pu)(x)=(Pu_x)(x)=\chi(x)f(x)=f(x)$ and $u(x)=u_x(x)=\chi(x)u_0(x)=u_0(x)$ if $x\in\Sigma$.
Hence $u$ solves the Cauchy problem~\eqref{eq:CauchyProblem}.
\end{proof}

\subsection{Stability for the Cauchy problem}
We conclude the discussion of the Cauchy problem for symmetric hyperbolic systems by showing stability.
This means that the solutions depend continuously on the data.
Note that if $u_0$ and $f$ have compact supports, then the solution $u$ of the Cauchy problem~\eqref{eq:CauchyProblem} has spacially compact support by \cref{cor:CPuniqueness}.

\begin{prop}[Stability of the Cauchy problem]\label{prop:CPstabitlity}
Let $P$ be a symmetric hyperbolic system over the globally hyperbolic manifold $M$.
Let $\Sigma\subset M$ be a smooth spacelike Cauchy hypersurface.

Then the map $\Cooc(M,E)\times \Cooc(\Sigma,E) \to \Coosc(M,E)$ mapping $(f,u_0)$ to the solution $u$ of the Cauchy problem \eqref{eq:CauchyProblem} is continuous.
\end{prop}

\begin{proof}
The map $\P:\Coo(M,E) \to \Coo(M,E)\times\Coo(\Sigma,E)$, $u\mapsto (Pu,u|_\Sigma)$, is linear and continuous.
Fix a compact subset $A\subset M$.
Then $\CooA(M,E)\times\Coo_{A\cap\Sigma}(\Sigma,E)$ is a closed subset of $\Coo(M,E)\times\Coo(\Sigma,E)$ and thus $V_A:=\P^{-1}(\CooA(M,E)\times\Coo_{A\cap\Sigma}(\Sigma,E))$ is a closed subset of $\Coo(M,E)$.
In particular, $\CooA(M,E)\times\Coo_{A\cap\Sigma}(\Sigma,E)$ and $V_A$ are Fr\'echet spaces.
By \cref{cor:CPuniqueness} and \tref{thm:CPexistence}, $\P$ maps $V_A$ bijectively onto $\CooA(M,E)\times\Coo_{A\cap\Sigma}(\Sigma,E)$.
The open mapping theorem for Fr\'echet spaces tells us that $(\P|_{V_A})^{-1}:\CooA(M,E)\times\Coo_{A\cap\Sigma}(\Sigma,E) \to V_A$ is continuous.
Now $V_A\subset\Coo(M,E)$ and $\Coo_{J(A)}(M,E)\subset\Coo(M,E)$ carry the relative topologies and $V_A\subset \Coo_{J(A)}(M,E)$ by \cref{cor:finitespeed}.
Thus the embeddings $V_A \hookrightarrow \Coo_{J(A)}(M,E)\hookrightarrow \Coosc(M,E)$ are continuous.
Hence the solution operator for the Cauchy problem yields a continuous map $\CooA(M,E)\times\Coo_{A\cap\Sigma}(\Sigma,E) \to \Coosc(M,E)$ for every compact $A\subset M$.
Therefore it is continuous as a map $\Cooc(M,E)\times \Cooc(\Sigma,E) \to \Coosc(M,E)$.
\end{proof}

\begin{rem}
\cref{cor:CPuniqueness}, \tref{thm:CPexistence} and \pref{prop:CPstabitlity} are often summarized by saying that the Cauchy problem \eqref{eq:CauchyProblem} is \emph{well posed}.
\end{rem}

\subsection{Green-hyperbolicity of symmetric hyperbolic systems}
Finally, we show that symmetric hyperbolic systems over globally hyperbolic manifolds are Green hyperbolic.

\begin{thm}
Symmetric hyperbolic systems over globally hyperbolic manifolds are Green hyperbolic.
\end{thm}

\begin{proof}
Let $P$ be a symmetric hyperbolic system over the globally hyperbolic manifold $M$.
Let $t:M\to\R$ be a Cauchy temporal function.
Put $I:=t(M)$.

We construct an advanced Green's operator for $P$.
Let $f\in\Cooc(M,E)$.
Choose $t_0\in I$ such that $K:=\supp f \subset I^+(\Sigma_{t_0})$.
We solve the Cauchy problem $Pu=f$ with initial condition $u|_{\Sigma_{t_0}}=0$.
Now put $G_+f:=u$.

This definition does not depend on the particular choice of $t_0$.
Namely, let $t_1<t_2$ be two values in $I$ such that $K\subset I^+(\Sigma_{t_i})$.
Then the solution of $Pu=f$ and $u|_{\Sigma_{t_1}}=0$ vanishes on $t^{-1}([t_1,t_2])$ because of \eqref{eq:finitefuturepast}.
Hence it coincides with the solution of the Cauchy problem $Pu=f$ with initial condition $u|_{\Sigma_{t_2}}=0$.

Thus we obtain a linear map $G_+:\Cooc(M,E)\to\Coo(M,E)$ such that $PG_+f=f$ for every $f\in\Cooc(M,E)$.
If $f=Pv$ for some $v\in\Cooc(M,E)$, then $u=v$ is the unique solution to the Cauchy problem $Pu=f$ with $u|_{\Sigma_{t_0}}=0$.
This shows $G_+Pv=v$ for every $v\in\Cooc(M,E)$.

By \eqref{eq:finitefuturepast}, $\supp(G_+f) \subset J^+(\supp f)$.
Hence $G_+$ is an advanced Green's operator.
A retarded Green's operator is constructed similarly by choosing $t_0\in I$ such that $K\subset I^-(\Sigma_{t_0})$.

Finally, $-\Pt$ is again a symmetric hyperbolic system and therefore has Green's operators.
Thus $\Pt$ has Green's operators and $P$ is Green hyperbolic.
\end{proof}

\subsection{Locally covariant quantum field theory}
In \cite[Thm.~3.10]{BG12b} we showed that Green-hyperbolic operators always give rise to bosonic locally covariant quantum field theories in the sense of \cite{BFV03}.
Fermionic locally covariant quantum field theories are much harder to construct.
In \cite[Thm.~3.20]{BG12b} it was shown that a construction is possible for formally selfadjoint Green-hyperbolic operators of first order if they are of \emph{positive type}, see \cite[Def.~3.12]{BG12b}.
One easily sees that if $Q$ is formally selfadjoint, then $Q$ is of positive type if and only if $P=iQ$ is a symmetric hyperbolic system.
In \cite{BG12b} we assumed that $Q$ is Green hyperbolic.
Here we have shown that this is actually automatic.

It is remarkable that formally selfadjoint symmetric hyperbolic systems give rise to both, bosonic and fermionic quantum field theories.
This shows that there is no spin-statistics theorem on the level of observable algebras.
One has to complement the observables by suitable states as in \cite{V01}.

In \cite{BG12b} it was shown that some but not all Dirac-type operators are of positive type.
The classical Dirac operator acting on spinor fields is of positive type.
In contrast, Buchdahl operators which describe higher spin fields are not. 
Therefore the theory of symmetric hyperbolic systems does not apply to Buchdahl operators;
nevertheless, they are Green hyperbolic.

%%%%%%%%%%%%%%%%%%%%%%%%%%%%%%%%%%%%%%%%%%%%%%%%%%%%%%%%%%%%%%%%%%%%%%%%%%%%%%%%%%%%%%%%%%%%%%%%%

\end{document}